\documentclass[reprint, superscriptaddress]{revtex4-1}
\usepackage[colorlinks=true, citecolor=red, urlcolor=blue ]{hyperref}
\usepackage{epsfig,newlfont,amssymb,amsfonts,amsmath,bm,subfigure,palatino,mathtools,amsthm,braket,times,soul,enumitem,color}
\usepackage[normalem]{ulem}
\usepackage{xcolor}
\usepackage{cancel}
\usepackage[T1]{fontenc}
\usepackage[normalem]{ulem}
\usepackage{graphicx}
\usepackage{placeins}
\usepackage{hyperref}
\usepackage{multirow}
\usepackage{hhline,booktabs}

\newtheorem{prop}{Proposition}

\usepackage{blindtext}
\usepackage{graphicx}
\usepackage{amsmath}
\usepackage{bm}
\usepackage{mathrsfs} 
\usepackage{hyperref}
\usepackage{geometry}
\usepackage{amsthm}

\usepackage{physics}
 \definecolor{darkgreen}{rgb}{0.47,0.67,0.19}
\newcommand{\wjy}[1]{{\color{black}#1}}
\newcommand{\ab}[0]{\color{black}}

\begin{document}

\title{Fidelity and Entanglement of Random Bipartite Pure States: Insights and Applications}
\author{George Biswas}
% \email{georgebsws@gmail.com}
\affiliation{Department of Physics, Tamkang University, Tamsui Dist., New Taipei 25137, Taiwan, ROC}
\affiliation{Center for Advanced Quantum Computing, Tamkang University, Tamsui Dist., New Taipei 25137, Taiwan, ROC}
%\affiliation{Department of Physics, National Institute of Technology Sikkim, Ravangla, South Sikkim 737 139, India}

\author{Shao-Hua Hu}
\affiliation{Department of Physics, Tamkang University, Tamsui Dist., New Taipei 25137, Taiwan, ROC}
% \affiliation{Center for Advanced Quantum Computing, Tamkang University, 151 Yingzhuan Rd, Tamsui Dist., New Taipei 25137, Taiwan, ROC}

\author{Jun-Yi Wu}
\email{junyiwuphysics@gmail.com}
\affiliation{Department of Physics, Tamkang University, Tamsui Dist., New Taipei 25137, Taiwan, ROC}
\affiliation{Center for Advanced Quantum Computing, Tamkang University, Tamsui Dist., New Taipei 25137, Taiwan, ROC}
\affiliation{Physics Division, National Center for Theoretical Sciences, Taipei 10617, Taiwan, ROC}

\author{Debasish Biswas}
% \email{debasish.abc@gmail.com}}
\affiliation{National Institute of Technology Jamshedpur, Adityapur, Jamshedpur, Jharkhand 831 014, India}

\date{\today}

\author{Anindya Biswas}
 \email{anindya@nitsikkim.ac.in}
\affiliation{Department of Physics, National Institute of Technology Sikkim, Ravangla, South Sikkim 737 139, India}

\begin{abstract}
We investigate the fidelity of Haar random bipartite pure states from a fixed reference quantum state and their bipartite entanglement. By plotting the fidelity and entanglement on perpendicular axes, we observe that the resulting plots exhibit non-uniform distributions. The distribution depends on the entanglement of the fixed reference quantum state used to quantify the fidelity of the random pure bipartite states. We find that the average fidelity of typical random pure bipartite qubits within a narrow entanglement range with respect to a randomly chosen fixed bipartite qubit is $\frac{1}{4}$. Extending our study to higher dimensional bipartite qudits, we find that the average fidelity of typical random pure bipartite qudits with respect to a randomly chosen fixed bipartite qudit remains constant within a narrow entanglement range. The values of these constants are \(\frac{1}{d^2}\), with d being the dimension of the local Hilbert space of the bipartite qudit system, suggesting a consistent relationship between entanglement and fidelity across different dimensions.
The probability distribution functions of fidelity with respect to a product state are analytically studied and used as a reference for the benchmarking of distributed quantum computing devices.

\end{abstract}

\maketitle

\section{Introduction}
\label{intro}

%Our 
The following study focuses on examining a collective property of typical random pure quantum states. When we use the term "typical," we refer to random quantum states that are uniformly distributed across the surface of a unit hypersphere, with the computational basis vectors serving as the axes. The set of typical random pure states is called the Haar uniformly distributed set of random pure states in the literature~\cite{Press92numericalrecipes,bengtsson_zyczkowski_2006,cohn2013measure,Dahlsten_2014,MISZCZAK2012118,doi:10.1063/1.3595693,e20100745,PhysRevA.93.032125,cohn2013measure,Zyczkowski_2001}. 
To express a general {bipartite qudit pure state} mathematically, we represent it as a sum of basis vectors, as shown in Eq.~\ref{general_pure_state}.
\begin{equation} \label{general_pure_state}
    |{\psi}\rangle=\frac{1}{C}\sum_{k=1}^{d^2}(a_{k}+ib_{k})|k\rangle
\end{equation}
In this equation, the coefficients \(a_{k}\) and \(b_{k}\) represent the real and imaginary parts, respectively, of the complex coefficient of the basis vector \(|k\rangle\), and \(C = \sqrt{\sum_{k=1}^{d^2}(a_k^2 + b_k^2)}\) is the normalization constant. We generate the random pure states Haar uniformly. 

Haar uniformity is a property that can be achieved by selecting the coefficients \(a_{k}\) and \(b_{k}\) independently from Gaussian distributions with vanishing mean and finite variance. Several references, including \cite{Press92numericalrecipes,bengtsson_zyczkowski_2006,cohn2013measure,Dahlsten_2014,MISZCZAK2012118,doi:10.1063/1.3595693,e20100745,PhysRevA.93.032125,cohn2013measure,Zyczkowski_2001,Biswas_2021,biswas2022spread}, discuss the utilization of Haar uniformity in different contexts. The probability density function of Gaussian distribution is given by Eq.~\ref{gaussian_int}, where \(\mu_G\) represents the mean and \(\sigma_G\) represents the standard deviation of the distribution.
\begin{equation}
\label{gaussian_int}
f_G(x)=\frac{1}{\sigma_G\sqrt{2\pi}}e^{-\frac{1}{2}\left(\frac{x-\mu_G}{\sigma_G}\right)^2}
\end{equation}
It is worth mentioning that when the individual variables are independently distributed according to Gaussian distributions, the joint probability distribution of these variables is spherically symmetric.
Note that the states are not initially distributed uniformly over the hypersphere's surface. Instead, they are distributed in the hyperspace according to Gaussian distributions. However, after the normalization process, these states become uniformly distributed over the hypersphere's surface with a unit radius in the same hyperspace~\cite{Biswas_2021}. 

In this study, we %quantify 
calculate the fidelity of typical random bipartite pure quantum states with respect to a fixed reference pure quantum state. We also calculate the %level of 
bipartite entanglement %exhibited by 
of these random states.

Fidelity is a fundamental concept in quantum information theory and is commonly used to compare the similarity of quantum states, assess the performance of quantum operations or channels, and quantify the accuracy of quantum state preparation and measurement~\cite{Yuan_2017,PhysRevLett.73.3047,Winter_2001}. The fidelity between two pure quantum states \(|\psi\rangle\) and \(|\phi\rangle\) quantifies the closeness of the two states and is defined as the absolute value square of their inner product~\cite{Yuan_2017,PhysRevLett.73.3047,LI2015158,doi:10.1080/09500349414552171,Winter_2001}. 
Mathematically, the fidelity \(F\) of two pure states is given by: 
\begin{equation}
    F(|\psi\rangle,|\phi\rangle)=|\langle\psi|\phi\rangle|^2.
\end{equation}
%The fidelity reaches its maximum value of 1 when the two states are identical, indicating perfect overlap. Conversely, a fidelity value 0 implies complete orthogonality between the states, indicating no overlap~\cite{Yuan_2017,PhysRevLett.73.3047,LI2015158,doi:10.1080/09500349414552171,Winter_2001}.

To quantify the entanglement of typical random bipartite pure quantum states, we have used entanglement entropy \(E\). %Entanglement entropy is a measure used to quantify bipartite entanglement of a bipartite pure state. 
For a bipartite pure state $|{\psi_{ab}}\rangle$ of a composite system $a$ and $b$, the entanglement entropy is defined as the von Neumann entropy of the reduced density matrix $\rho_a = Tr_b(|\psi_{ab}\rangle\langle\psi_{ab}|)$: 
\begin{equation}
    E(|\psi_{ab}\rangle)=-Tr(\rho_a\log\rho_a).
\end{equation}
The logarithm in the equation is taken to the base 2 for bipartite qubits and taken to the base ${d}$ for bipartite ${d}$-dimensional qudits~\cite{Janzing2009}. %\wjy{\sout{Note that the total Hilbert space dimension for bipartite qubits is \(2^2=4\) and similarly the total Hilbert space dimension of bipartite \(\widetilde{d}\)-dimensional qudits \({\widetilde{d}}^{2}\)}}.

%Existing research has shown that when considering Haar uniformly distributed random pure states in a d-dimensional Hilbert space, their 
{\ab It is known that the average mutual fidelity of Haar uniformly distributed random pure states in a \(d\)-dimensional Hilbert space converges to $\frac{1}{d}$}~\cite{PhysRevA.71.032313}. This convergence also holds when we focus on a single reference state and compute the average fidelity between this fixed state and all other states randomly distributed according to the Haar measure. 
In this paper, we look into the insights of the fidelity distribution and address the question, "What is the average fidelity of a subset of Haar uniformly distributed random bipartite qubits/qudits, defined by a fixed entanglement entropy value, with respect to a fixed reference bipartite qubit/qudit?"
% \textcolor{black}{Note that our primary objective is to investigate the probability distributions of fidelity and their average, conditioned by the entanglement entropy value of Haar uniformly distributed bipartite pure states. The analysis of probability distributions of the entanglement of Haar uniform pure states has been thoroughly explored in previous literature~\cite{Hayden2006,Fukuda2014,Hayashi2017,Biswas_2021}.}
\wjy{Note that the analysis of probability distributions of the entanglement of Haar uniform pure states has been thoroughly explored in previous literature~\cite{Hayden2006,Fukuda2014,Hayashi2017,Biswas_2021}.  Our primary objective of this paper is to investigate the probability distributions of fidelity and their average, conditioned by the entanglement entropy value of Haar uniformly distributed bipartite pure states. }

\section{Average fidelity and entanglement of typical bipartite qubits/qudits}
\label{fetbq}

This section is dedicated to quantifying the fidelity between a fixed pure bipartite qubit/qudit and typical random bipartite pure qubits/qudits. Furthermore, we calculate  
 the average fidelity \((\textnormal{F}_\textnormal{avg})\) of these typical random pure qubits/qudits with respect to the fixed reference bipartite qubit/qudit by sorting them according to their %level 
 {\ab degree} of bipartite entanglement.

\begin{figure}[!ht]
\centering
   \includegraphics[width=\linewidth]{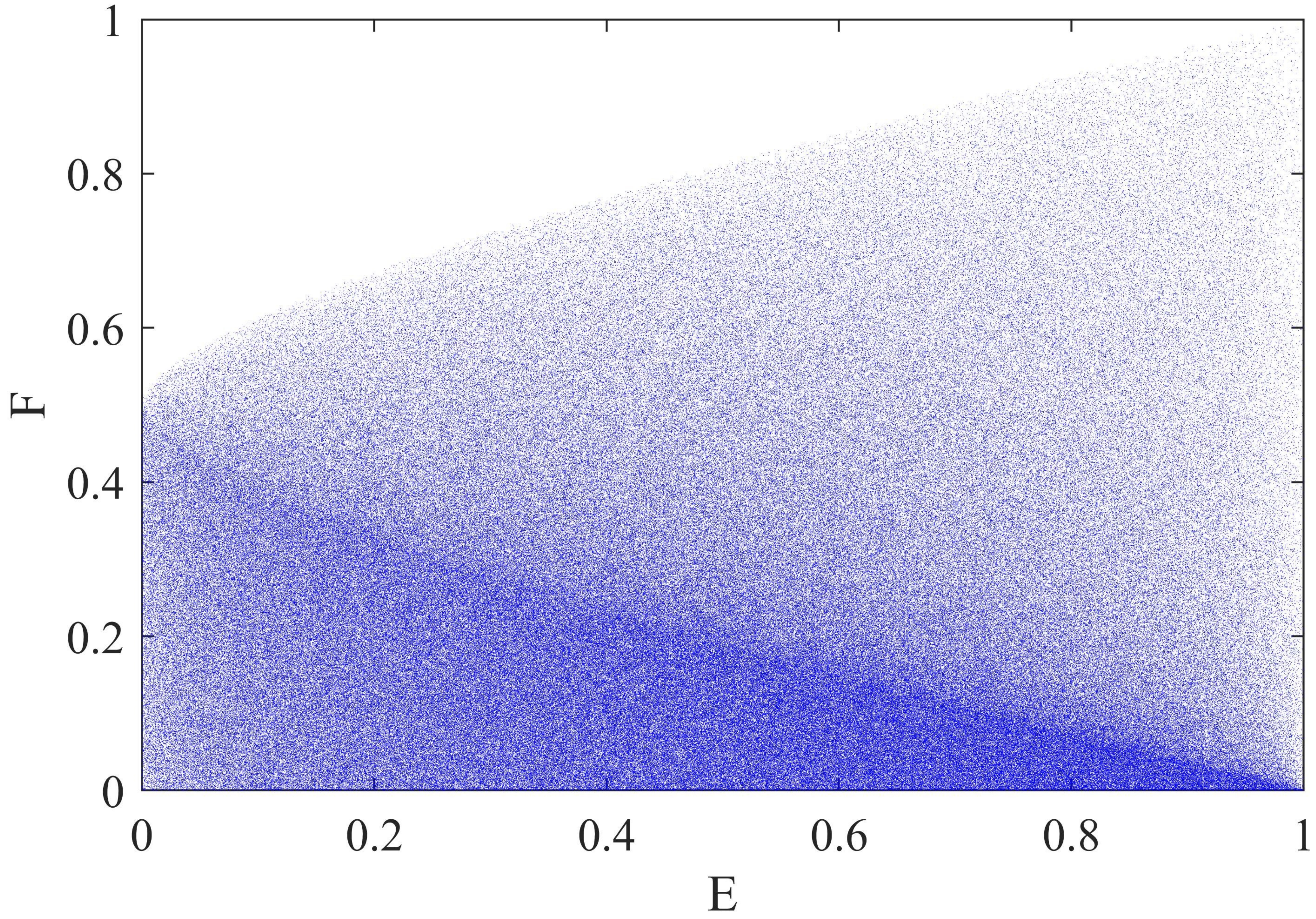}
  \caption{Scatter plot of the entanglement entropy and the fidelity of Haar uniform random {two-qubit pure states} with respect to a fixed {two-qubit maximally entangled state}. The entanglement entropy of the random {two-qubit states} is plotted along the horizontal axis and fidelity between the fixed state and the random states is plotted along the vertical axis.}
\label{scatter_p_bell}
\end{figure}

\begin{figure}[!ht]
\centering
   \includegraphics[width=\linewidth]{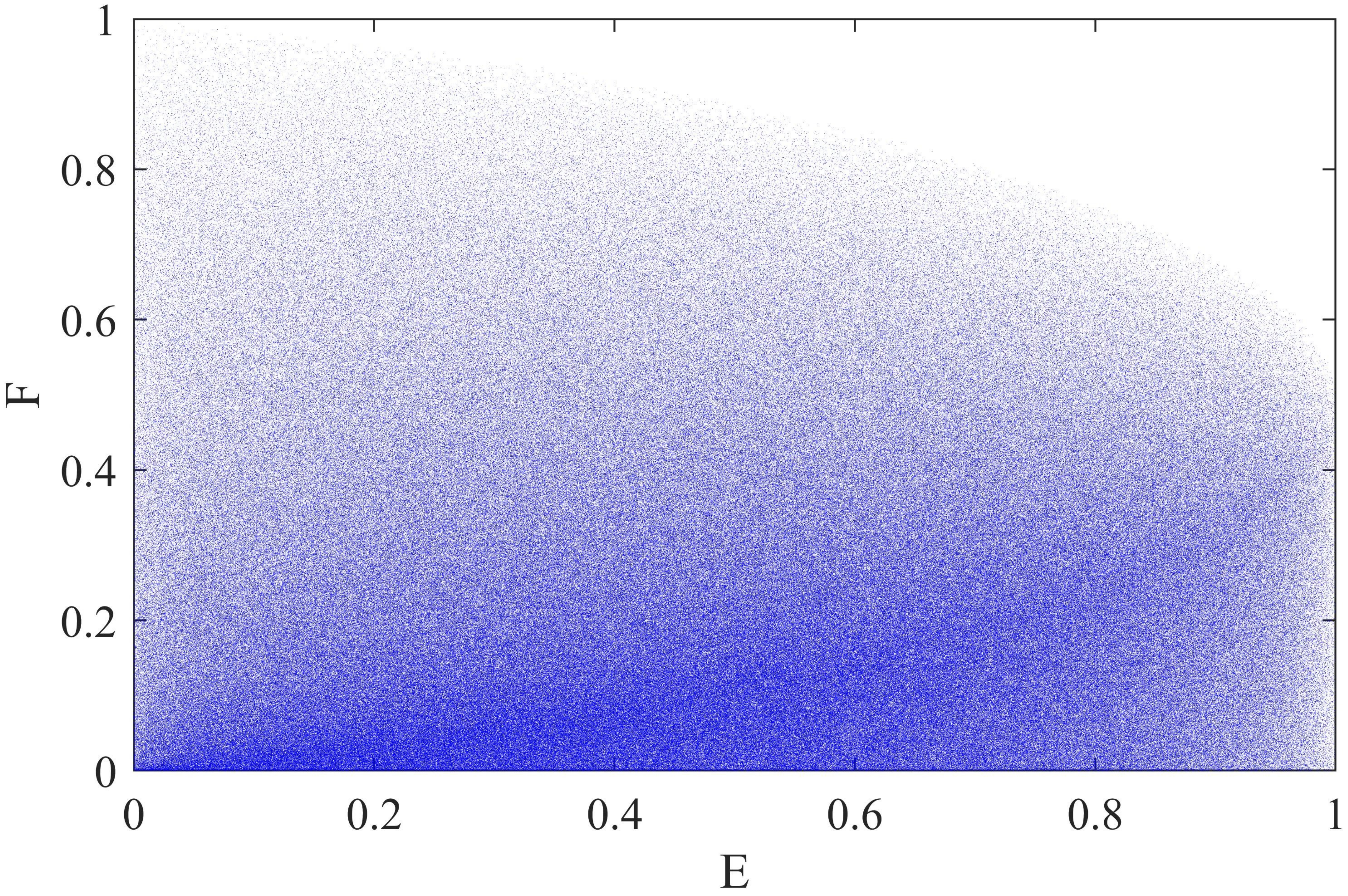}
  \caption{Scatter plot of the entanglement entropy and the fidelity of Haar uniform random {two-qubit pure states} with respect to a fixed {two-qubit separable pure state}. The entanglement entropy of the random states is plotted along the horizontal axis and the fidelity between the fixed state and the random states is plotted along the vertical axis.}
\label{scatter_p_separable}
\end{figure}

%The fidelity distribution depends on the entanglement of the chosen reference quantum state. 
Fig.~\ref{scatter_p_bell} displays a scatter plot of fidelity (F) against entanglement entropy (E) when a maximally entangled state is selected as the fixed reference bipartite pure qubit, while Fig.~\ref{scatter_p_separable} shows the same plot when a separable bipartite pure qubit is chosen as the fixed reference quantum state.
\textcolor{black}{Reference~\cite{PhysRevA.71.032313} established a probability distribution for the fidelity of d-dimensional Haar uniform random pure states as ${P}(f) = (d-1)(1-f)^{d-2}$, irrespective of the fixed reference quantum state. In this paper, we examine the detailed fidelity distribution in relation to the entanglement of bipartite pure states, we see from}
%\sout{As shown in} 
Fig.~\ref{scatter_p_bell} and Fig.~\ref{scatter_p_separable}, \textcolor{black}{that} the fidelity distributions of Haar uniformly distributed bipartite qubits vary based on the entanglement of the fixed reference qubits used for quantifying fidelity. 
It is important to note that if we selectively pick Haar uniform pure bipartite qubits based on their entanglement entropy (equivalent to selecting specific vertical lines in Fig.\ref{scatter_p_bell}), these subsets exhibit different fidelity distributions. Additionally, the fidelity distribution of a subset selected by a entanglement entropy value will also differ when the fixed reference state used for quantifying fidelity %is of 
{\ab has a} different {\ab degree of bipartite} entanglement. In other words, two narrow vertical lines at the same position on the horizontal axis in Fig.\ref{scatter_p_bell} and Fig.~\ref{scatter_p_separable} will have distinct fidelity distributions.

Moreover, from a mathematical perspective, it becomes evident that the fidelity distribution of these subsets of Haar random pure bipartite qubits remains consistent even when different fixed reference qubits with the same entanglement are employed. This principle extends to higher-dimensional bipartite qudits {\ab as well}.
The fidelity distribution for Haar uniformly distributed random {two-qudit pure states} (represented as \(P_F\)),  with a specific entanglement entropy value \(\epsilon\) %concerning 
{\ab with reference to} a fixed reference bipartite qudit \(|\phi_{ab}\rangle\)  is described by Eq.~\ref{qubit_fdy_delta}.
% \begin{align}\label{qubit_fdy_delta}
% P_F(f,|\phi_{ab}\rangle,\epsilon) &= \int_{|\psi_{ab}\rangle} \dd |\psi_{ab}\rangle \ \delta(\abs{\langle\phi_{ab}|\psi_{ab}\rangle}^2-f)\notag \\
% &\hspace{0.5cm}\delta(E(|\psi_{ab}\rangle)-\epsilon),
% \end{align}
\begin{align}\label{qubit_fdy_delta}
%P_F(f,|\phi_{ab}\rangle,\epsilon) &=
P_F(f,\epsilon) &=\int_{\text{Haar}} \dd \psi_{ab} \ \delta(\abs{\langle\phi_{ab}|\psi_{ab}\rangle}^2-f)\notag \\
&\hspace{0.5cm}\delta(E(|\psi_{ab}\rangle)-\epsilon),
\end{align}
%where {\ab \(f\) is a fidelity value between 0 and 1,} \(|\psi_{ab}\rangle\) is one of the Haar uniformly distributed random {two-qudit states} and the post-selection of the states corresponding to a particular entanglement entropy value \(\epsilon\) is represented by the Dirac delta term \(\delta(E(|\psi_{ab}\rangle)-\epsilon)\).
\textcolor{black}{where the fidelity $f$ is treated as a random variable and the delta functions act to pinpoint specific bipartite states $|\psi_{ab}\rangle$ from the Haar random distribution based on their fidelity $f$ relative to $|\phi_{ab}\rangle$ and the desired entanglement level $\epsilon$ during integration over the Haar measure~\cite{Santanu_2008}. Note that the Haar measure is normalized such that $\int_{\text{Haar}}d\psi_{ab}=1$.}
Now, let us insert the identity operator as \(I=(U_a\otimes U_b)^{\dagger}(U_a\otimes U_b)\) in Eq.~\ref{qubit_fdy_delta}, where \((U_a\otimes U_b)\) is an arbitrary local unitary operator. Then we have,
% \begin{align}
% P_F(f,|\phi_{ab}\rangle,\epsilon) &= \int_{|\psi_{ab}\rangle} \dd |\psi_{ab}\rangle \ \notag \\
% &\hspace{0.5cm}\delta(\abs{\langle\phi_{ab}|(U_a\otimes U_b)^{\dagger}(U_a\otimes U_b)|\psi_{ab}\rangle}^2\notag \\
% &\hspace{0.5cm}-f)\ \delta(E(|\psi_{ab}\rangle)-\epsilon)
% \notag\\
% &= \int_{|\psi_{ab}\rangle} \dd |\psi_{ab}\rangle \ \notag \\
% &\hspace{0.5cm} \delta(\abs{\langle\phi_{ab}|(U_a\otimes U_b)^{\dagger}|\psi'_{ab}\rangle}^2-f)\ \notag \\
% &\hspace{0.5cm} \delta(E(|\psi'_{ab}\rangle)-\epsilon), 
% \notag
% \end{align}
% where, \(|\psi'_{ab}\rangle=(U_a\otimes U_b)|\psi'_{ab}\rangle\). Note that \(E(|\psi'_{ab}\rangle) = E((U_a\otimes U_b)|\psi_{ab}\rangle)=E(|\psi_{ab}\rangle)\) because the local unitary operator does not change the entanglement of a bipartite quantum state. 
\begin{align}
%&P_F(f,|\phi_{ab}\rangle,\epsilon)  \notag \\
&P_F(f,\epsilon)  \notag \\
= & \int_{\text{Haar}} \dd \psi_{ab} \, \delta\left(\abs{\langle\phi_{ab}|(U_a\otimes U_b)^{\dagger}(U_a\otimes U_b)|\psi_{ab}\rangle}^2-f\right)
\notag\\
& \times \delta(E(|\psi_{ab}\rangle)-\epsilon)
\end{align}
Since the local unitary operator $U_a\otimes U_b$ does not change the entanglement of a bipartite quantum state, we may substitute $|\psi_{ab}\rangle$ by $|\psi'_{ab}\rangle=(U_a\otimes U_b)|\psi_{ab}\rangle$, 
\begin{align}
%& P_F(f,|\phi_{ab}\rangle,\epsilon)  \notag \\
& P_F(f,\epsilon)  \notag \\
= & \int_{\text{Haar}} \dd \psi'_{ab} \, 
\delta\left(\abs{\langle\phi_{ab}|(U_a\otimes U_b)^{\dagger}|\psi'_{ab}\rangle}^2-f\right)
\notag\\
& \times \delta(E(|\psi'_{ab}\rangle)-\epsilon) 
\end{align}
%
% Also, since Haar uniform distribution is invariant under a unitary transformation we may take the integration over \(|\psi'_{ab}\rangle\). Therefore, we have,
% \begin{align}\label{qubit_invariant}
% P_F(f,|\phi_{ab}\rangle,\epsilon) &= \int_{|\psi'_{ab}\rangle} \dd |\psi'_{ab}\rangle \ \notag \\
% &\hspace{0.5cm}\delta(\abs{\langle\phi_{ab}|(U_a\otimes U_b)^{\dagger}|\psi'_{ab}\rangle}^2-f)\ \notag \\
% &\hspace{0.5cm} \delta(E(|\psi'_{ab}\rangle)-\epsilon), 
% \notag\\
% &= P_F(f,(U_a\otimes U_b)|\phi_{ab}\rangle,\epsilon)
% \end{align}
% \replace{
% We see from Eq.~\ref{qubit_invariant} that the choice of the fixed reference qudit \(|\phi_{ab}\rangle\) is arbitrary up to local unitary transformations.}
It means that the fidelity distribution of the two-qudit state with an entanglement entropy \(\epsilon\) does not depend on the choice of the reference two-qudit state \(|\phi_{ab}\rangle\) %is arbitrary 
up to local unitary transformations.
\begin{equation}
    %P_F(f,|\phi_{ab}\rangle,\epsilon) =
    %P_F(f,\epsilon) = P_F(f,(U_a\otimes U_b)|\phi_{ab}\rangle,\epsilon).
[P_F(f,\epsilon)]_{{\color{black}|\phi_{ab}\rangle}} = [P_F(f,\epsilon)]_{{\color{black}(U_a\otimes U_b)|\phi_{ab}\rangle}},
    \label{Equal_Under_Local_operation}
\end{equation}
\textcolor{black}{where the subscripts of the fidelity distributions indicate the reference two-qubit states used to quantify fidelity.}

Notably, a subset of Haar uniformly distributed random bipartite qubits, defined by a fixed entanglement entropy value, exhibits an average fidelity remarkably close to \(\frac{1}{4}\), regardless of the specific fixed reference qubit. Expanding our study to higher-dimensional bipartite qudits, we observe that the average fidelity of random pure bipartite qudits remains constant at \(\frac{1}{d^2}\)  within a narrow entanglement range. This suggests a consistent relationship between entanglement and fidelity across different dimensions. Fig.~\ref{avg_fdy_qudits} {\ab is a visual representation of }%visually represents 
these {\ab numerical} observations. Note that we find a scarcity of states with high and low entanglement entropy in higher dimensions during the numerical preparation of Haar uniform random bipartite qudits.

\begin{figure}[!ht]
\centering
   \includegraphics[width=\linewidth]{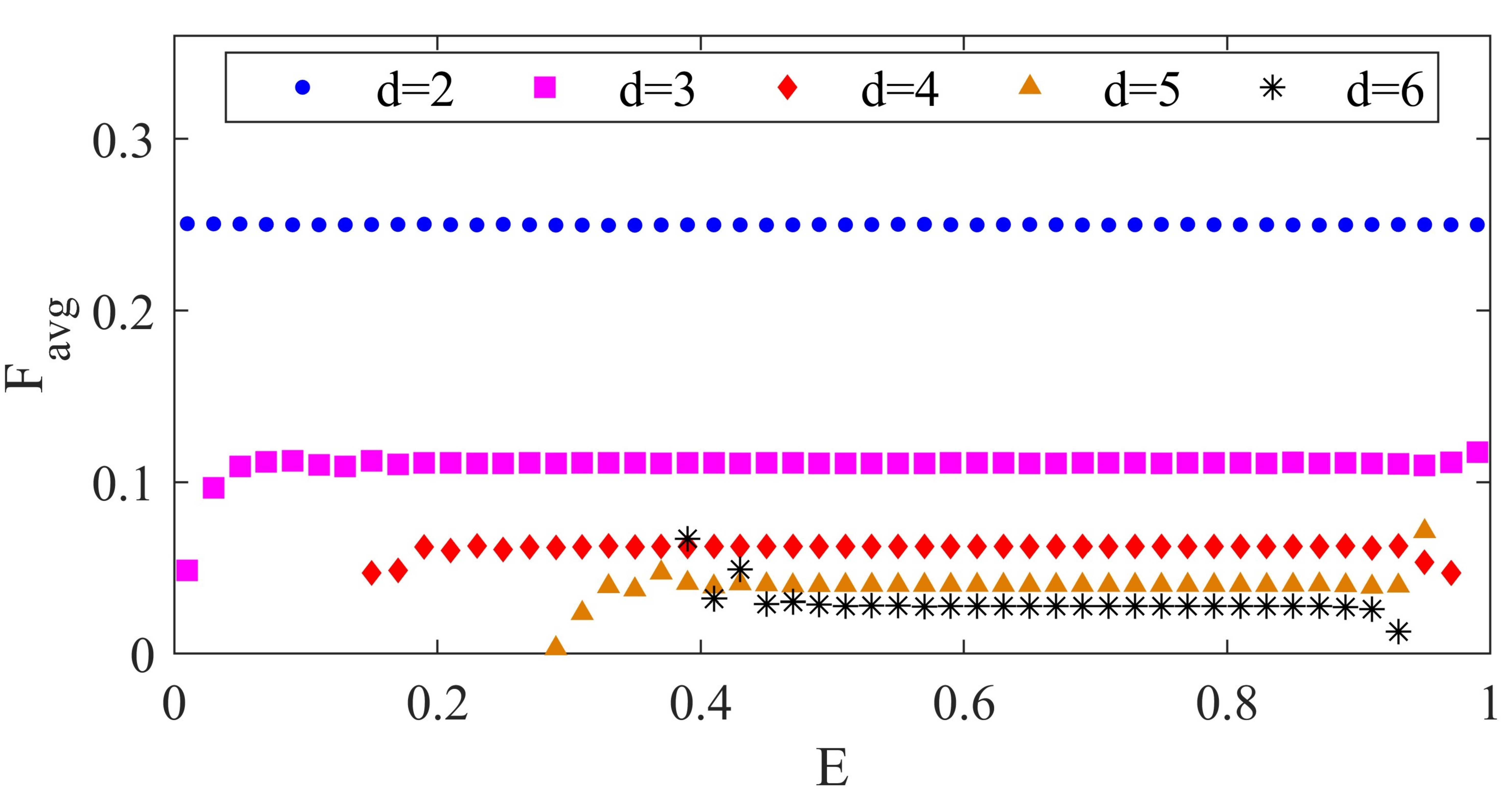}
  \caption{We plot the average fidelity (\({\textnormal{F}_{\textnormal{avg}}}\)) of the subsets of Haar uniformly chosen random {bipartite qudit pure states subject to} a particular value of entanglement entropy (E) with respect to a fixed {reference state for different local dimensionality \({d}\).}
  There are 50 windows of entanglement entropy values along the horizontal axis and the width of the windows is 0.02.}
\label{avg_fdy_qudits}
\end{figure}

Mathematical analysis further supports these findings. We can find the
average fidelity of Haar uniformly distributed bipartite {qudit states} from a fixed reference {state} mathematically. We may write any bipartite pure states by Schmidt decomposition, and entanglement entropy
may be calculated using the square of Schmidt coefficients. We may show
that the average fidelity of bipartite states having the same Schmidt
coefficients with respect to a fixed reference state is
constant at \(\frac{1}{d^2}\) as follows. 
Consider a bipartite qudit state
\begin{equation}
  \left|I_{\boldsymbol{s}} \right\rangle  = \sum_{n = 1}^{d}s_n\left| \left. n_{a} \right\rangle \right.\otimes\left|n_{b} \right\rangle  
\end{equation}
% \(\left|I_{P} \right\rangle  = \sum_{n = 1}^{\wjy{d}}\sqrt{P_{n}}\left| \left. n_{a} \right\rangle \right.\bigotimes\left|n_{b} \right\rangle\)
with a specific set of Schmidt coefficients
\(\boldsymbol{s} = \left\{ s_{1}, s_{2}, \ldots, s_{d} \right\}\), where $\boldsymbol{s}$ is ordered by $s_{1} < s_{2} < \cdots < s_{d}$. %{\ab not clear \wjy{[wjy: the ordering is now clarified]}}
We denote the set of states generated from local Haar random unitaries \(U_a\otimes U_b\) acting on the state \(|I_{\boldsymbol{s}}\rangle\) as
\textcolor{black}{
\begin{equation}
\label{eq::s_coeff_set}
    \mathbb{S}_{\boldsymbol{s}}:=\left\{ U_a\otimes U_b\left| I_{\boldsymbol{s}} \right\rangle
    : U_a, U_b \in \mathbb{U}(d)
    \right\}, 
\end{equation}
where $\mathbb{U}(d)$ is the set of $d$-dimensional unitaries.}
All the states \(|\psi_{\boldsymbol{s}}\rangle\) in \(\mathbb{S}_{\boldsymbol{s}}\) have the same Schmidt coefficients \(\boldsymbol{s}\) and hence the same entanglement entropy 
\(E\left( \left| \left. I_{\boldsymbol{s}} \right\rangle \right. \right) = \sum_{n = 1}^{d} s_{n}^2\, \log s_{n}^2 = \epsilon\).
The sets \(\mathbb{S}_{\boldsymbol{s}}\) of different Schmidt coefficients \(\boldsymbol{s}\) are distinct subsets of the set of Haar random bipartite qudit states. 
\begin{equation}
    \mathbb{S}_{\boldsymbol{s}}\cap\mathbb{S}_{\boldsymbol{s}'}
    =
    \emptyset, \text{for all } \boldsymbol{s}\neq\boldsymbol{s}'
\end{equation}
The union of the distinct subsets \(\mathbb{S}_{\boldsymbol{s}}\) over the Schmidt coefficients that have the same entanglement entropy $\epsilon$ forms the whole set of the target ensemble of bipartite qudit states with the entanglement entropy $\epsilon$.
\begin{equation}
    \mathbb{E}_{\epsilon}:=
    \{|\psi_{ab}\rangle: E(|\psi_{ab}\rangle) = \epsilon  \}
    =
    \bigcup_{\boldsymbol{s}:E(\boldsymbol{s})=\epsilon}\mathbb{S}_{\boldsymbol{s}}
\end{equation}
Suppose we fix a reference state
\begin{equation}
    \left|I_{\boldsymbol{s}'} \right\rangle  = \sum_{m = 1}^{d}s'_{m}\left|m_{a} \right\rangle\bigotimes\left| m_{b} \right\rangle
\end{equation}
with the Schmidt coefficients
\(\boldsymbol{s}' = \left\{ s'_{1},s'_{2} ,\ldots,s'_{d} \right\}\).
The average fidelity $F_{avg,\boldsymbol{s}'}(\mathbb{E}_{\epsilon})$ of the Haar uniformly distributed states in 
\(\mathbb{E}_{\epsilon}\) with respect to the fixed reference state
\(\left| I_{\boldsymbol{s}'} \right\rangle \) is
\textcolor{black}{given by 
\begin{align}
\label{eq::ave_fid_ent_entropy}
    F_{avg,\boldsymbol{s}'}(\mathbb{E}_{\epsilon})
    & =
    \int%_{Haar}
    \dd f \ f \ P_F(f|\epsilon) 
    \notag \\
    & = \int%_{Haar}
    \dd f \ f \ \frac{P_F(f,\epsilon)}{P(\epsilon)}, 
\end{align} 
where, $P_F(f|\epsilon)$ is a conditional probability of fidelity with the condition on entanglement entropy value, $P_F(f,\epsilon)$ is given by Eq.~\ref{qubit_fdy_delta} and $P(\epsilon)=\int_{\text{Haar}} \dd \psi_{ab} \ \delta(E(|\psi_{ab}\rangle)-\epsilon)$. The average fidelity $F_{avg,\boldsymbol{s}'}(\mathbb{E}_{\epsilon})$ is then
\begin{align}
\label{eq::ave_fid_ent_entropy}
    F_{avg,\boldsymbol{s}'}(\mathbb{E}_{\epsilon})
    = 
    \int_{Haar(\mathbb{E}_{\epsilon})} \dd \psi_{ab}\ |\langle I_{\boldsymbol{s'}}|\psi_{ab} \rangle|^2,
\end{align}
where $\int_{Haar(\mathbb{E}_{\epsilon})} \dd \psi_{ab} \,g(\psi_{ab})$ is the integral of a function $g(\psi_{ab})$ over the Haar measure normalized within the set $\mathbb{E}_{\epsilon}$,
\begin{align}
\label{eq::normalized_Haar_measure}
    & \int_{Haar(\mathbb{E}_{\epsilon})} 
    \dd \psi_{ab} \, g(\psi_{ab})
    \notag\\
    := &
    \frac{1}{P(\epsilon)}
    \int_{Haar} \dd \psi_{ab}\        \delta(E(\ket{\psi_{ab}}-\epsilon)) \, g(\psi_{ab}).
\end{align}
Since the set $\mathbb{E}_{\epsilon}$ is a distinct union of $\mathbb{S}_{\boldsymbol{s}}$, the average fidelity $F_{avg,\boldsymbol{s}'}(\mathbb{E}_{\epsilon})$ within the set $\mathbb{E}_{\epsilon}$ is an average of the average fidelity $F_{avg,\boldsymbol{s}'}(\mathbb{S}_{\boldsymbol{s}})$ within the set $\mathbb{S}_{\boldsymbol{s}}$ over all the Schmidt coefficients $\boldsymbol{s}'$ that has an entanglement entropy of $\epsilon$, }
\begin{align}\label{equivalence}
    F_{avg,\boldsymbol{s}'}(\mathbb{E}_{\epsilon})
    =
    \int_{\boldsymbol{s}:E(\boldsymbol{s})=\epsilon}
    \dd \boldsymbol{s} \, p_{\boldsymbol{s}}
    F_{avg,\boldsymbol{s}'}(\mathbb{S}_{\boldsymbol{s}}).
\end{align} 
Here, the probability distribution function $p_{\boldsymbol{s}}$ of the Schmidt coeffiient $\boldsymbol{s}$ within the set of the entanglement entropy $\epsilon$ is given by 
\begin{equation}
    p_{\boldsymbol{s}}
    =
    \frac{
        \int_{Haar} \dd \psi_{ab}\; \delta( \boldsymbol{S}(|\psi_{ab}\rangle) -\boldsymbol{s} )
    }{
        \int_{Haar} \dd \psi_{ab}\; \delta( E(|\psi_{ab}\rangle) -\epsilon )
    },
\end{equation}
where $\boldsymbol{S}(|\psi_{ab}\rangle)$ is the Schmidt coefficient of the state $|\psi_{ab}\rangle$ in the ascending order.
\textcolor{black}{The average fidelity $F_{avg,\boldsymbol{s}'}(\mathbb{S}_{\boldsymbol{s}})$ of the Haar uniformly distributed states in 
$\mathbb{S}_{\boldsymbol{s}}$ with respect to the fixed reference state
\(\left| I_{\boldsymbol{s}'} \right\rangle \) within the set $\mathbb{S}_{\boldsymbol{s}}$  is
given by 
\begin{align}
\label{eq::ave_fid_ent_entropy}
    F_{avg,\boldsymbol{s}'}(\mathbb{S}_{\boldsymbol{s}})
    & =
    \int_{Haar(\mathbb{S}_{\boldsymbol{s}})}
    \dd \psi_{ab}\ |\langle I_{\boldsymbol{s'}}|\psi_{ab} \rangle|^2,
\end{align} 
where $\int_{Haar(\mathbb{S}_{\boldsymbol{s}})}    \dd \psi_{ab} $ is the integral over the normalized Haar measure within the set $\mathbb{S}_{\boldsymbol{s}}$ analogous to Eq. \ref{eq::normalized_Haar_measure}.
According to Eq. \ref{eq::s_coeff_set}, its value is determined by }
\begin{align}
    & F_{avg,\boldsymbol{s}'}(\mathbb{S}_{\boldsymbol{s}})
    \notag \\
    % & =
    % \iint_U  \dd U_a \dd U_b
    % |\langle I_{\boldsymbol{s}'}| U_a\otimes U_b| I_{\boldsymbol{s}}\rangle|^2
    % \notag \\
    = &
    \iint_{\mathbb{U}(d)^{\otimes 2}} \dd U_a \dd U_b
    \;\langle I_{\boldsymbol{s}'}| 
    U_a\otimes U_b
    | I_{\boldsymbol{s}}\rangle
    \langle I_{\boldsymbol{s}}|
    U_a^{\dagger}\otimes U_b^{\dagger}
    | I_{\boldsymbol{s}'}\rangle.
\end{align}
According to Proposition 5 in \cite{10.1063/5.0038838}, a linear map $\Lambda$ passing through a quantum channel of Haar random unitary is a completely depolarizing channel multiplied with the trace of $\Lambda$,
\begin{equation}
    \int_{\mathbb{U}(d)} dU \; U \Lambda U^{\dagger} 
    = \frac{tr(\Lambda)}{d} \mathbb{I}_d.
\end{equation}
The state $\ket{I_{\boldsymbol{s}}}$ passing the Haar random unitary will be then completely depolarized,
\begin{align}
    \int_{U(d)^{\otimes 2}} \dd U_a \dd U_b \;
    U_a\otimes U_b  \ket{I_{\boldsymbol{s}}}\bra{I_{\boldsymbol{s}}}
    U_a^{\dagger}\otimes U_b^{\dagger} 
    = \frac{1}{d^2}\mathbb{I}_{d^2}.
\end{align}
This leads to a constant average fidelity of Haar uniformly distributed random  bipartite pure states with particular Schmidt coefficients relative to a fixed
reference state, 
\begin{equation}
    F_{avg,\boldsymbol{s}'}(\mathbb{S}_{\boldsymbol{s}}) = \frac{1}{d^2}.
\end{equation}
As a result of Eq.~\eqref{equivalence},
%\eqref{eq::ave_fid_ent_entropy}, 
the average fidelity of Haar uniformly distributed random  bipartite pure states with a specific entanglement entropy $\epsilon$ relative to a fixed reference state is also a constant
\begin{equation}
    F_{avg,\boldsymbol{s}'}(\mathbb{E}_{\epsilon}) = \frac{1}{d^2}.
\end{equation}
This mathematical result aligns with our numerical observations, \textcolor{black}{depicted in Fig.~\ref{avg_fdy_qudits}, where $F_{avg,\boldsymbol{s}'}(\mathbb{E}_{\epsilon})\equiv F_{avg}$.}

\section{Probability distribution of fidelity of bipartite states with respect to a pure product state and its application}
\label{PDF}

Previously, we observed that a specific subset of Haar uniformly distributed random bipartite qudits, characterized by a constant entanglement entropy, demonstrates an average fidelity of $\frac{1}{d^2}$ with respect to any given bipartite reference qudit. Despite maintaining uniform average fidelity across various subsets defined by distinct entanglement entropy values, the probability distributions associated with these subsets vary. Moreover, the probability distributions differ when considering different fixed reference qudits of different entanglement entropy.
In the subsequent analysis, we delve into the calculation of probability density functions (PDFs) for fidelity distributions. 

First, we focus on subsets of Haar uniformly distributed random bipartite qubits defined by different entanglement entropy %concerning 
{\ab with respect to} a fixed pure product state as reference. {\ab We find that the probability density functions of the fidelity distributions corresponding to  any entanglement based subsets of Haar uniform states, are functions of  entanglement entropy.} %This exploration unfolds as a function of the entanglement entropy values assigned to different subsets {\ab not clear}. 
In essence, we are calculating the probability density functions for fidelity distributions along different narrow vertical columns of Fig.~\ref{scatter_p_separable}, as a function of their position on the horizontal axis.
In appendix~\ref{apxA}, we derived the closed-form expression for the PDF of fidelity for a subset of Haar random bipartite qubits with respect to a fixed reference pure product state. The subset is defined by fixed Schmidt coefficients $\left(\sqrt{\frac{1+\gamma}{2}},\sqrt{\frac{1-\gamma}{2}}\right)$, with a parameter $\gamma\in\,[0,1\,]$, or a fixed entanglement entropy $E = - \left[\left(\frac{1+\gamma}{2}\right) \log_2{\left(\frac{1+\gamma}{2}\right)}  + \left(\frac{1-\gamma}{2}\right) \log_2{\left(\frac{1-\gamma}{2}\right)}\right]$. \textcolor{black}{We denote the PDF of fidelity with respect to a separable state conditional on the parameter $\gamma$ as $P_{\text{sep}}(f|\gamma)$ and} the obtained relation is given by Eq.~\ref{pdf_pfp}. 
\begin{align}
%F_{\gamma}=
%\cancel{{P}(f|\gamma)}
\textcolor{black}{P_{\text{sep}}(f|\gamma)} =
%P_{\gamma}(f) =
\begin{cases} 
\frac{1}{\gamma}\ln{\frac{1+\gamma}{1-\gamma}}\hspace{0.45cm} \text{for} \ \ 0 \le f\le \frac{1-\gamma}{2}\\
\frac{1}{\gamma}\ln{\frac{1+\gamma}{2f}}\hspace{0.45cm} \text{for} \ \frac{1-\gamma}{2}\le f\le \frac{1+\gamma}{2} \\ 
0\hspace{1.53cm}\text{for} \ \frac{1+\gamma}{2} \le f\le 1 
\end{cases}
\label{pdf_pfp}
\end{align}
Fig.~\ref{F_gamma} depicts the PDF $\textcolor{black}{P_{\text{sep}}(f|\gamma)}$ for various values of $\gamma$.
%It is noteworthy that for various values of $\gamma$, the distributions share a common mean value of $\frac{1}{4}$. However, their standard deviation, as given by Eq.~\ref{std}, and skewness, as given by Eq.~\ref{skw}, exhibit distinct characteristics.
%\begin{align}
%    Std(\gamma) & = \sqrt{\int_0^1{(f-0.25)^2\ \mathscr{P}(f,\gamma)\ \dd f}} 
%     = \frac{\sqrt{4\gamma^2 +3}}{12}
%    \label{std}
%\end{align}
%\begin{align}
%    Skw(\gamma) & = \frac{\int_0^1{(f-0.25)^3\ \mathscr{P}(f,\gamma)\ \dd f}}{\left(Std(\gamma)\right)^3} 
%     = \frac{{18\gamma^2}}{(4\gamma^2 +3)^{\frac{3}{2}}}
%    \label{skw}
%\end{align}
\begin{figure}[!ht]
\centering
   \includegraphics[width=\linewidth]{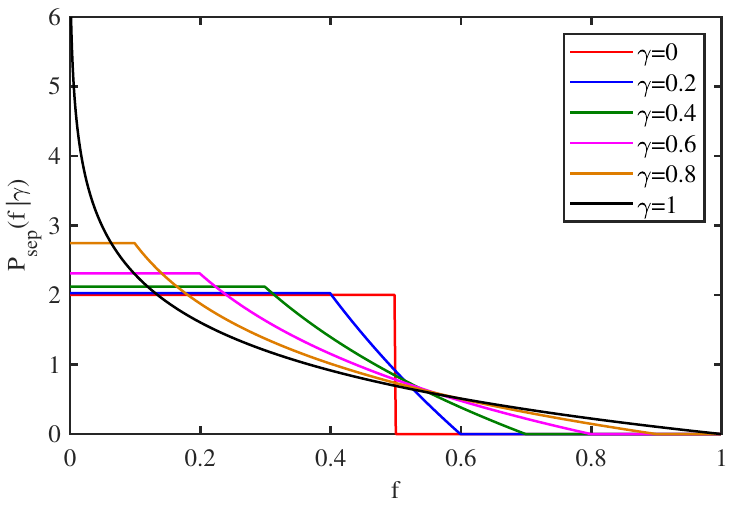}
  \caption{Probability distribution function (PDF) of fidelity for a sub-set of Haar random bipartite qubits with respect to a fixed reference pure product state. The subsets are defined by various entanglement values or various Schmidt coefficient parameters $\gamma$.}
\label{F_gamma}
\end{figure}

Second, We extend our analysis for higher dimensional qudits. In Appendix~\ref{apxA}, we present the derivation of the PDF for the fidelity of Haar uniformly distributed maximally entangled qudits. This distribution is computed with respect to a fixed separable qudit as a reference. 
\textcolor{black}{We denote the PDF of fidelity with respect to a separable state for the maximally entangled subset of Haar uniform states as $P_{\text{sep}}(f|\text{max. ent.})$ and} the PDF for $d$-dimensional qudits is expressed by Eq.~\ref{qudit_pdf}.
\begin{align}
%F_{d}=\mathscr{P}(f|d) = 
%P_{d}(f\textcolor{black}{\cancel{|d}}) = 
%P_{d}(f) 
%\cancel{{P}(f|d)}
&\textcolor{black}{P_{\text{sep}}(f|\text{max. ent.})} 
\notag\\
=&
\begin{cases} 
d(d-1)(1-df)^{d-2}\hspace{0.2cm} \text{for} \ \ 0 \le f\le \frac{1}{d}\\
0\hspace{3.0cm}\text{for} \hspace{0.13cm}\frac{1}{d} \le f\le 1 
\end{cases}
\label{qudit_pdf}
\end{align}
Fig.~\ref{F_d} depicts the PDF $\textcolor{black}{P_{\text{sep}}(f|\text{max. ent.})}$ for various dimension $d$.

\begin{figure}[!ht]
\centering
   \includegraphics[width=\linewidth]{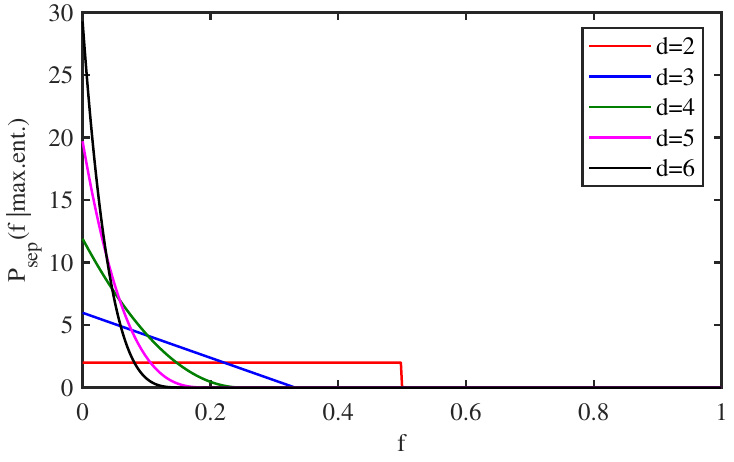}
  \caption{Probability distribution function (PDF) for the fidelity of Haar uniformly distributed maximally entangled bipartite qudits with respect to a fixed separable qudit as a reference for various dimensional bipartite qudits.}
\label{F_d}
\end{figure}

%\subsection{Benchmarking entanglement assisted local operation}
The PDFs $  \textcolor{black}{P_{\text{sep}}(f|\text{max. ent.})}$ and $ \textcolor{black}{P_{\text{sep}}(f|\gamma)}$ for fidelity distributions can serve as references to benchmark entanglement-assisted local operations. Entanglement-assisted local operations, coupled with classical communication, constitutes a crucial component in distributed quantum computation (DQC)~\cite{PhysRevA.62.052317,Wu2023entanglement,https://doi.org/10.1049/iet-qtc.2020.0002,andresmartinez2023distributing}. It is an important device for DQC, illustrated schematically in Fig.~\ref{EALO}.
\begin{figure}[!ht] 
\centering
\includegraphics[width=0.5\linewidth]{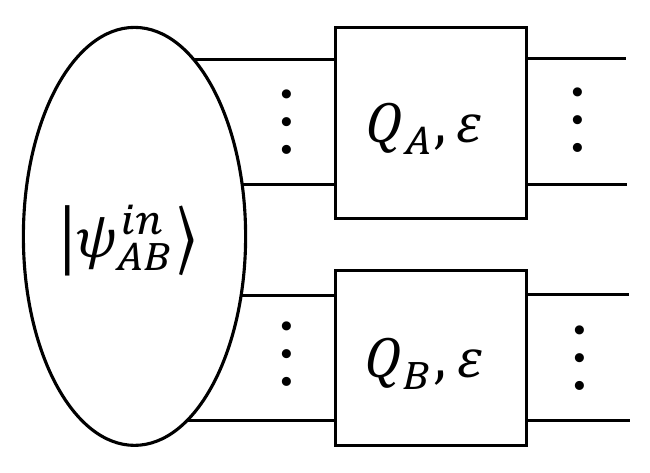}
\caption{%\wjy{[Comment: this figure can be improved with the new version used in QIP poster.]} 
Schematic visualization of the DQC device. Two quantum processing units $Q_A$ and $Q_B$ engage in local unitary operations on {bipartite entangled state} $\ket{\psi_{AB}^{in}}$ to facilitate distributed quantum computation.}
\label{EALO}
\end{figure}
Ideally, the DQC device operates with maximally entangled states as inputs. Eq.~\ref{qudit_pdf} describes the probability distribution of fidelity for maximally entangled input states %concerning 
{\ab with respect to} a fixed pure product state, considering Haar uniform local unitary operations on both parts of the bipartite states.

Subsequently, we consider a depolarizing error model, where both parts of the maximally entangled bipartite states
experience depolarization with an error probability, denoted as $\varepsilon$, during local unitary operations. The fidelity distribution concerning a fixed pure product state of the output from this error-prone DQC device varies with the error probability $\varepsilon$. We numerically computed the error-prone fidelity distribution and compared it with our error-free reference PDF $P_d$ from Eq.~\ref{qudit_pdf}. 
\begin{figure}[!ht]
\centering
   \includegraphics[width=\linewidth]{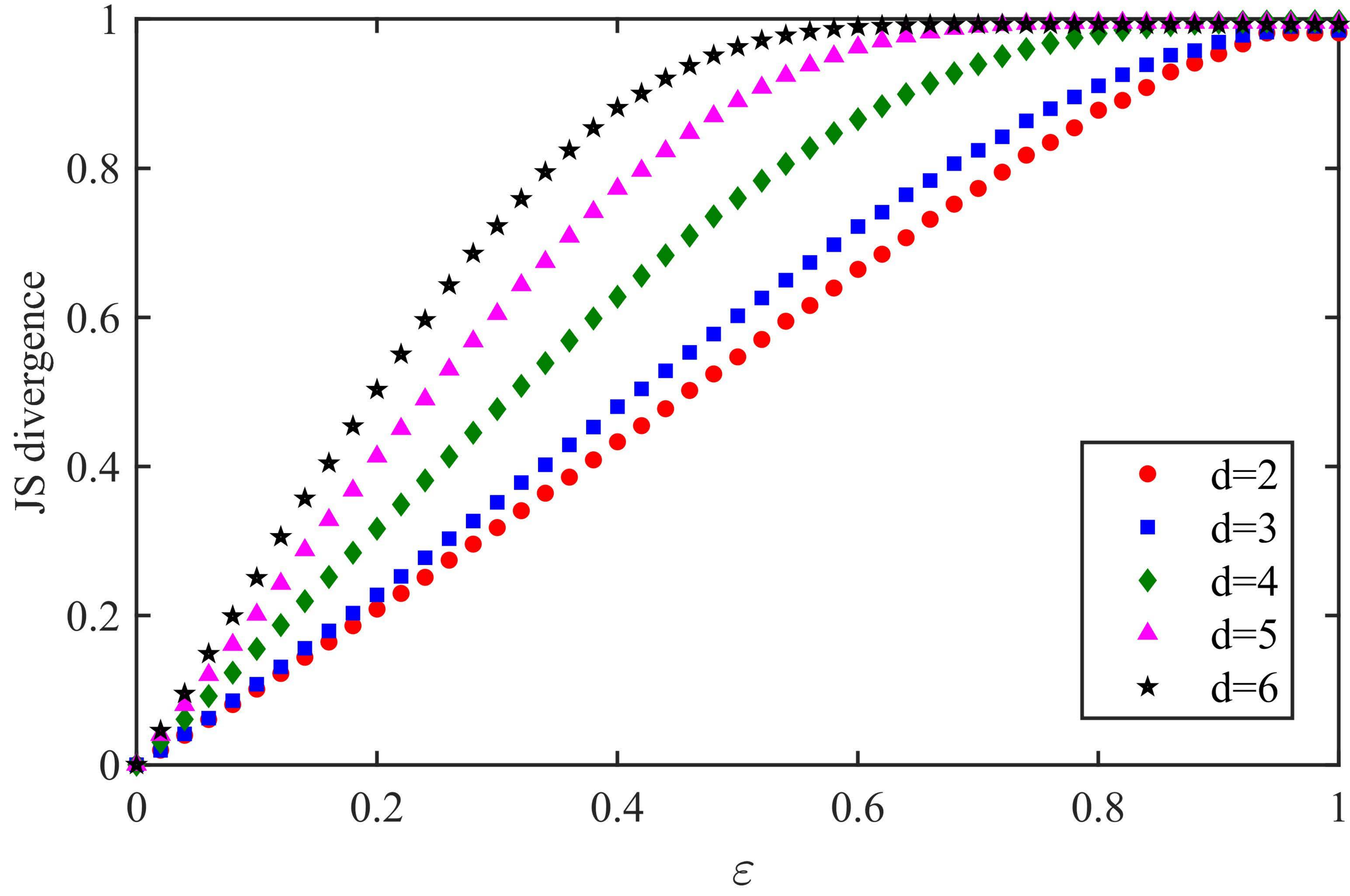}
  \caption{{Jensen-Shannon (JS) divergence} between the error-prone and error-free PDFs as a function of the error probability $\varepsilon$. The analysis is conducted for different dimensional maximally entangled bipartite qudits, showcasing the variation in divergence between the two distributions with changing error rates $\varepsilon$.}
\label{JSD_EALO}
\end{figure}
To quantify the difference between the two probability distributions, we employed the Jensen-Shannon (JS) divergence ~\cite{61115}. The {JS divergence} between probability distributions $P$ and $Q$ is defined as:
\(
D_{JS}(P \,||\, Q) = \frac{1}{2} \left( D_{\text{KL}}(P \,||\, M) + D_{\text{KL}}(Q \,||\, M) \right)
\),
where $M=\frac{P+Q}{2}$ is the midpoint distribution, and $D_{\text{KL}}$ denotes the Kullback-Leibler (KL) divergence.
The KL divergence between probability distributions \( P \) and \( M \) is defined as:
\(
D_{\text{KL}}(P \,||\, M) = \sum_{i} P(i) \log_2\left(\frac{P(i)}{M(i)}\right)
\).
In Fig.~\ref{JSD_EALO}, the {JS divergence} between the error-prone and error-free PDFs of maximally entangled bipartite qudits is depicted with respect to the error probability $\varepsilon$, considering various dimensional bipartite qudits. These divergences serve as effective benchmarks for assessing the error probability $\varepsilon$ associated with the DQC device. We further observed that the DQC device is more sensitive to error in higher dimensions.

In the context of bipartite qubits, we also possess the reference PDFs for partially entangled input states denoted by $  \textcolor{black}{P_{\text{sep}}(f|\gamma)}$ as outlined in Eq.~\ref{pdf_pfp}. In Fig.~\ref{JSD_qubit}, the JS divergence between the error-prone and error-free PDFs of fidelity of bipartite qubits is depicted with respect to the error probability $\varepsilon$, considering input states of various entanglement. We observed that the DQC device is more sensitive to error for input states of higher entanglement.

%In appendix~\ref{apxA} we derive analytically a depolarizing error model with error probability $\varepsilon=1-r$ and with partially entangled qubits as inputs, expressed by Eq.~\ref{pdf_ed}. Fig.~\ref{JSD_qubit} illustrate the JS divergence between the PDFs of Eq.~\ref{pdf_ed} and~\ref{pdf_pfp} for various entanglement entropy or corresponding Schmidt coefficient parameter $\gamma$. \replace{These divergences can be employed as a randomized benchmarking method for a source of partially entangled bipartite qubits.}{}

\begin{figure}[!ht]
\centering
   \includegraphics[width=\linewidth]{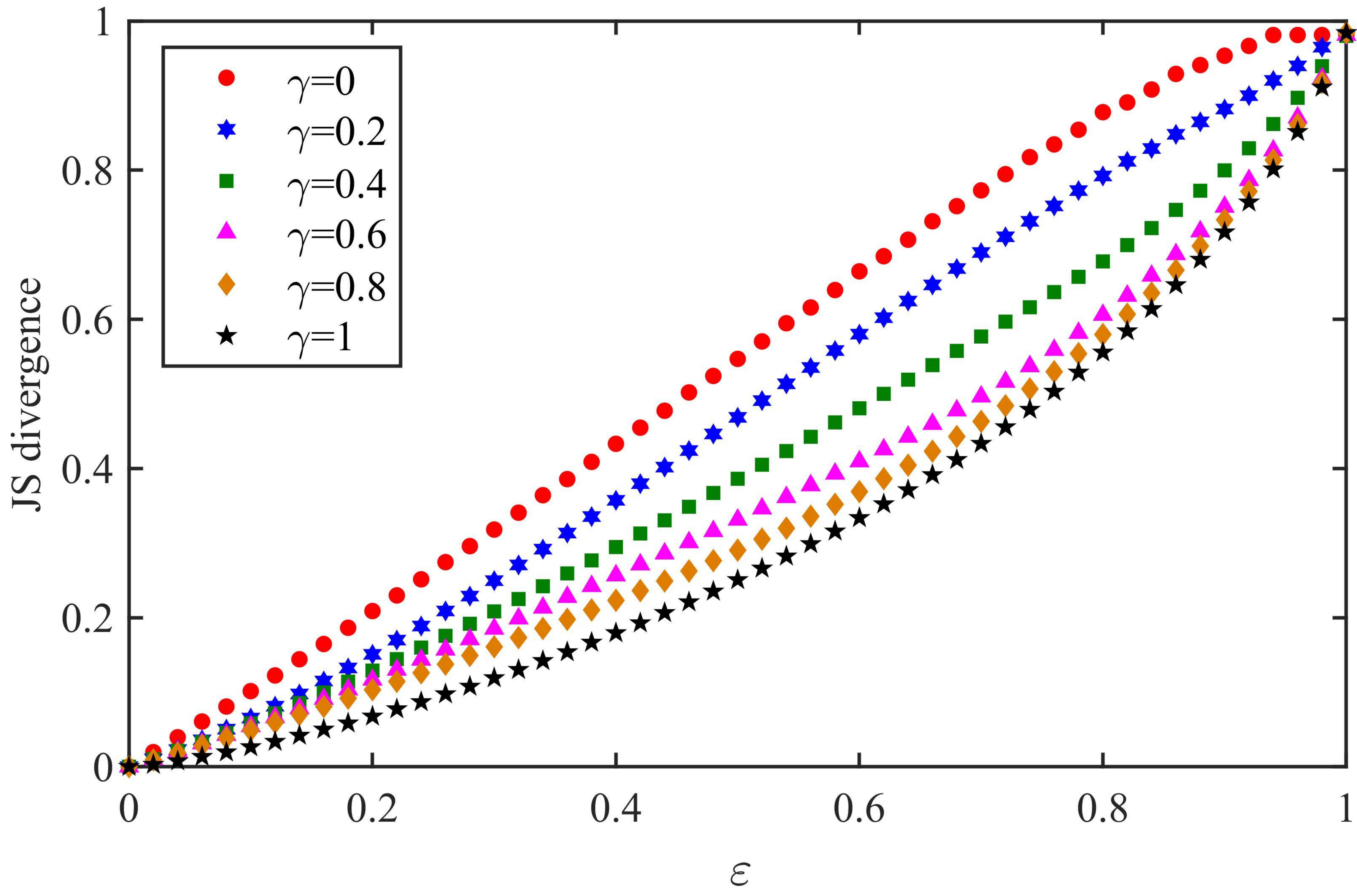}%JSD_epep
  \caption{Jensen-Shannon (JS) divergence between the error-prone and error-free PDFs of fidelity of bipartite qubits for input states of various entanglement entropy or corresponding Schmidt coefficient parameter $\gamma$. Showcasing the variation in divergence between two distributions with changing error rates $\varepsilon$.}
\label{JSD_qubit}
\end{figure}

%\begin{widetext}
%\begin{align}
%F_{r, \gamma}=\mathscr{P}(f|r,\gamma) &= \begin{cases} 
%0 \hspace{2.885cm} \text{for} \ \ 0 \le f \le \frac{(1-\gamma r)(1-r)}{4} \\
%\frac{1}{\gamma r^2}\ln{\frac{4f}{(1-\gamma r)(1-r)}}\hspace{0.45cm} \text{for} \ \frac{(1-\gamma r)(1-r)}{4} \le f \le \frac{(1+\gamma r)(1-r)}{4}\\
%\frac{1}{\gamma r^2}\ln{\frac{1+\gamma r}{1-\gamma r}}\hspace{1.38cm} \text{for} \ \frac{(1+\gamma r)(1-r)}{4}\le f \le \frac{(1-\gamma r)(1+r)}{4} \\ \frac{1}{\gamma r^2}\ln{\frac{(1+r)(1+\gamma r)}{4f}}\hspace{0.45cm} \text{for} \ \frac{(1-\gamma r)(1+r)}{4}\le f \le \frac{(1+\gamma r)(1+r)}{4}
%\\ 0\hspace{2.885cm} \text{for\ others} %\ \frac{(1+\gamma r)(1+r)}{4} \le f \le 1
%\end{cases}
%\label{pdf_ed}
%\end{align}
%\end{widetext}

%\subsection{Benchmarking of rank-2 DQC device}

\section{Discussion}
\label{discussion}

We investigated the fidelity and entanglement properties of typical random bipartite pure states. We analyzed the fidelity of these states %concerning 
{\ab with respect to} a fixed reference quantum state and observed non-uniform distributions. We found that the distribution depends on the entanglement of the fixed reference quantum state used for quantifying fidelity. Interestingly, we discovered that the average fidelity of typical random pure bipartite qubits consistently %falls at
{\ab remains} \(\frac{1}{4}\) within a narrow entanglement range, regardless of the fixed reference quantum state. {Extending} our study to $(d\times d)$-dimensional {bipartite systems}, we found that the average fidelity of typical random {bipartite qudit pure states} within a narrow entanglement range remains constant at \(\frac{1}{d^2}\), suggesting a consistent relationship between entanglement and fidelity across different dimensions. Additionally, we demonstrate that the {\ab probability density functions} %PDFs $F_\gamma$ and $F_d$ 
for fidelity distributions can serve as references to benchmark %DQC 
{\ab distributed quantum computing} devices. Our findings provide valuable insights into the fidelity and entanglement properties of typical random bipartite pure states. These results deepen our understanding of the relationship between fidelity and entanglement, paving the way for future research and applications in quantum information theory. In conclusion, our study contributes to the ongoing exploration of random quantum states and their properties, providing a foundation for further investigations into the fundamental aspects of quantum mechanics and their practical implications in various fields.

\begin{acknowledgments}
GB, SHH, and JYW are supported by the National Science and Technology Council, Taiwan, under Grant No. NSTC 
111-2923-M-032-002-MY5, %NL-TW
111-2627-M-008-001, %NQT
112-2112-M-032-008-MY3, %QIO
112-2811-M-032-005-MY2, %PostDoc
112-2119-M-008-007. %NQT
\end{acknowledgments}

\bibliography{PhD_references}

\begin{thebibliography}{28}%
\makeatletter
\providecommand \@ifxundefined [1]{%
 \@ifx{#1\undefined}
}%
\providecommand \@ifnum [1]{%
 \ifnum #1\expandafter \@firstoftwo
 \else \expandafter \@secondoftwo
 \fi
}%
\providecommand \@ifx [1]{%
 \ifx #1\expandafter \@firstoftwo
 \else \expandafter \@secondoftwo
 \fi
}%
\providecommand \natexlab [1]{#1}%
\providecommand \enquote  [1]{``#1''}%
\providecommand \bibnamefont  [1]{#1}%
\providecommand \bibfnamefont [1]{#1}%
\providecommand \citenamefont [1]{#1}%
\providecommand \href@noop [0]{\@secondoftwo}%
\providecommand \href [0]{\begingroup \@sanitize@url \@href}%
\providecommand \@href[1]{\@@startlink{#1}\@@href}%
\providecommand \@@href[1]{\endgroup#1\@@endlink}%
\providecommand \@sanitize@url [0]{\catcode `\\12\catcode `\$12\catcode `\&12\catcode `\#12\catcode `\^12\catcode `\_12\catcode `\%12\relax}%
\providecommand \@@startlink[1]{}%
\providecommand \@@endlink[0]{}%
\providecommand \url  [0]{\begingroup\@sanitize@url \@url }%
\providecommand \@url [1]{\endgroup\@href {#1}{\urlprefix }}%
\providecommand \urlprefix  [0]{URL }%
\providecommand \Eprint [0]{\href }%
\providecommand \doibase [0]{http://dx.doi.org/}%
\providecommand \selectlanguage [0]{\@gobble}%
\providecommand \bibinfo  [0]{\@secondoftwo}%
\providecommand \bibfield  [0]{\@secondoftwo}%
\providecommand \translation [1]{[#1]}%
\providecommand \BibitemOpen [0]{}%
\providecommand \bibitemStop [0]{}%
\providecommand \bibitemNoStop [0]{.\EOS\space}%
\providecommand \EOS [0]{\spacefactor3000\relax}%
\providecommand \BibitemShut  [1]{\csname bibitem#1\endcsname}%
\let\auto@bib@innerbib\@empty
%</preamble>
\bibitem [{\citenamefont {Press}\ \emph {et~al.}(1992)\citenamefont {Press}, \citenamefont {Teukolsky}, \citenamefont {Vetterling},\ and\ \citenamefont {Flannery}}]{Press92numericalrecipes}%
  \BibitemOpen
  \bibfield  {author} {\bibinfo {author} {\bibfnamefont {W.~H.}\ \bibnamefont {Press}}, \bibinfo {author} {\bibfnamefont {S.~A.}\ \bibnamefont {Teukolsky}}, \bibinfo {author} {\bibfnamefont {W.~T.}\ \bibnamefont {Vetterling}}, \ and\ \bibinfo {author} {\bibfnamefont {B.~P.}\ \bibnamefont {Flannery}},\ }\href {https://www.cambridge.org/in/academic/subjects/mathematics/numerical-recipes/numerical-recipes-example-book-c-2nd-edition?format=PB&isbn=9780521437202} {\emph {\bibinfo {title} {Numerical Recipes in C: The Art of Scientific Computing. Second Edition}}}\ (\bibinfo  {publisher} {Cambridge University Press},\ \bibinfo {year} {1992})\BibitemShut {NoStop}%
\bibitem [{\citenamefont {Bengtsson}\ and\ \citenamefont {Zyczkowski}(2006)}]{bengtsson_zyczkowski_2006}%
  \BibitemOpen
  \bibfield  {author} {\bibinfo {author} {\bibfnamefont {I.}~\bibnamefont {Bengtsson}}\ and\ \bibinfo {author} {\bibfnamefont {K.}~\bibnamefont {Zyczkowski}},\ }\href {\doibase 10.1017/CBO9780511535048} {\emph {\bibinfo {title} {Geometry of Quantum States: An Introduction to Quantum Entanglement}}}\ (\bibinfo  {publisher} {Cambridge University Press},\ \bibinfo {year} {2006})\BibitemShut {NoStop}%
\bibitem [{\citenamefont {Cohn}(2013)}]{cohn2013measure}%
  \BibitemOpen
  \bibfield  {author} {\bibinfo {author} {\bibfnamefont {D.}~\bibnamefont {Cohn}},\ }\href {https://books.google.co.in/books?id=PEC3BAAAQBAJ} {\emph {\bibinfo {title} {Measure Theory: Second Edition}}}\ (\bibinfo  {publisher} {Springer New York},\ \bibinfo {year} {2013})\BibitemShut {NoStop}%
\bibitem [{\citenamefont {Dahlsten}\ \emph {et~al.}(2014)\citenamefont {Dahlsten}, \citenamefont {Lupo}, \citenamefont {Mancini},\ and\ \citenamefont {Serafini}}]{Dahlsten_2014}%
  \BibitemOpen
  \bibfield  {author} {\bibinfo {author} {\bibfnamefont {O.~C.~O.}\ \bibnamefont {Dahlsten}}, \bibinfo {author} {\bibfnamefont {C.}~\bibnamefont {Lupo}}, \bibinfo {author} {\bibfnamefont {S.}~\bibnamefont {Mancini}}, \ and\ \bibinfo {author} {\bibfnamefont {A.}~\bibnamefont {Serafini}},\ }\href {\doibase 10.1088/1751-8113/47/36/363001} {\bibfield  {journal} {\bibinfo  {journal} {Journal of Physics A: Mathematical and Theoretical}\ }\textbf {\bibinfo {volume} {47}},\ \bibinfo {pages} {363001} (\bibinfo {year} {2014})}\BibitemShut {NoStop}%
\bibitem [{\citenamefont {Miszczak}(2012)}]{MISZCZAK2012118}%
  \BibitemOpen
  \bibfield  {author} {\bibinfo {author} {\bibfnamefont {J.~A.}\ \bibnamefont {Miszczak}},\ }\href {\doibase https://doi.org/10.1016/j.cpc.2011.08.002} {\bibfield  {journal} {\bibinfo  {journal} {Computer Physics Communications}\ }\textbf {\bibinfo {volume} {183}},\ \bibinfo {pages} {118} (\bibinfo {year} {2012})}\BibitemShut {NoStop}%
\bibitem [{\citenamefont {Życzkowski}\ \emph {et~al.}(2011)\citenamefont {Życzkowski}, \citenamefont {Penson}, \citenamefont {Nechita},\ and\ \citenamefont {Collins}}]{doi:10.1063/1.3595693}%
  \BibitemOpen
  \bibfield  {author} {\bibinfo {author} {\bibfnamefont {K.}~\bibnamefont {Życzkowski}}, \bibinfo {author} {\bibfnamefont {K.~A.}\ \bibnamefont {Penson}}, \bibinfo {author} {\bibfnamefont {I.}~\bibnamefont {Nechita}}, \ and\ \bibinfo {author} {\bibfnamefont {B.}~\bibnamefont {Collins}},\ }\href {\doibase 10.1063/1.3595693} {\bibfield  {journal} {\bibinfo  {journal} {Journal of Mathematical Physics}\ }\textbf {\bibinfo {volume} {52}},\ \bibinfo {pages} {062201} (\bibinfo {year} {2011})}\BibitemShut {NoStop}%
\bibitem [{\citenamefont {Enríquez}\ \emph {et~al.}(2018)\citenamefont {Enríquez}, \citenamefont {Delgado},\ and\ \citenamefont {Życzkowski}}]{e20100745}%
  \BibitemOpen
  \bibfield  {author} {\bibinfo {author} {\bibfnamefont {M.}~\bibnamefont {Enríquez}}, \bibinfo {author} {\bibfnamefont {F.}~\bibnamefont {Delgado}}, \ and\ \bibinfo {author} {\bibfnamefont {K.}~\bibnamefont {Życzkowski}},\ }\href {\doibase 10.3390/e20100745} {\bibfield  {journal} {\bibinfo  {journal} {Entropy}\ }\textbf {\bibinfo {volume} {20}} (\bibinfo {year} {2018}),\ 10.3390/e20100745}\BibitemShut {NoStop}%
\bibitem [{\citenamefont {Singh}\ \emph {et~al.}(2016)\citenamefont {Singh}, \citenamefont {Zhang},\ and\ \citenamefont {Pati}}]{PhysRevA.93.032125}%
  \BibitemOpen
  \bibfield  {author} {\bibinfo {author} {\bibfnamefont {U.}~\bibnamefont {Singh}}, \bibinfo {author} {\bibfnamefont {L.}~\bibnamefont {Zhang}}, \ and\ \bibinfo {author} {\bibfnamefont {A.~K.}\ \bibnamefont {Pati}},\ }\href {\doibase 10.1103/PhysRevA.93.032125} {\bibfield  {journal} {\bibinfo  {journal} {Phys. Rev. A}\ }\textbf {\bibinfo {volume} {93}},\ \bibinfo {pages} {032125} (\bibinfo {year} {2016})}\BibitemShut {NoStop}%
\bibitem [{\citenamefont {Zyczkowski}\ and\ \citenamefont {Sommers}(2001)}]{Zyczkowski_2001}%
  \BibitemOpen
  \bibfield  {author} {\bibinfo {author} {\bibfnamefont {K.}~\bibnamefont {Zyczkowski}}\ and\ \bibinfo {author} {\bibfnamefont {H.-J.}\ \bibnamefont {Sommers}},\ }\href {\doibase 10.1088/0305-4470/34/35/335} {\bibfield  {journal} {\bibinfo  {journal} {Journal of Physics A: Mathematical and General}\ }\textbf {\bibinfo {volume} {34}},\ \bibinfo {pages} {7111} (\bibinfo {year} {2001})}\BibitemShut {NoStop}%
\bibitem [{\citenamefont {Biswas}\ \emph {et~al.}(2021)\citenamefont {Biswas}, \citenamefont {Biswas},\ and\ \citenamefont {Sen}}]{Biswas_2021}%
  \BibitemOpen
  \bibfield  {author} {\bibinfo {author} {\bibfnamefont {G.}~\bibnamefont {Biswas}}, \bibinfo {author} {\bibfnamefont {A.}~\bibnamefont {Biswas}}, \ and\ \bibinfo {author} {\bibfnamefont {U.}~\bibnamefont {Sen}},\ }\href {\doibase 10.1088/1367-2630/ac37c8} {\bibfield  {journal} {\bibinfo  {journal} {New Journal of Physics}\ }\textbf {\bibinfo {volume} {23}},\ \bibinfo {pages} {113042} (\bibinfo {year} {2021})}\BibitemShut {NoStop}%
\bibitem [{\citenamefont {Biswas}\ \emph {et~al.}(2023)\citenamefont {Biswas}, \citenamefont {Sarkar}, \citenamefont {Biswas},\ and\ \citenamefont {Sen}}]{biswas2022spread}%
  \BibitemOpen
  \bibfield  {author} {\bibinfo {author} {\bibfnamefont {G.}~\bibnamefont {Biswas}}, \bibinfo {author} {\bibfnamefont {S.}~\bibnamefont {Sarkar}}, \bibinfo {author} {\bibfnamefont {A.}~\bibnamefont {Biswas}}, \ and\ \bibinfo {author} {\bibfnamefont {U.}~\bibnamefont {Sen}},\ }\href {\doibase 10.1088/1367-2630/aced1e} {\bibfield  {journal} {\bibinfo  {journal} {New Journal of Physics}\ }\textbf {\bibinfo {volume} {25}},\ \bibinfo {pages} {083030} (\bibinfo {year} {2023})}\BibitemShut {NoStop}%
\bibitem [{\citenamefont {Yuan}\ and\ \citenamefont {Fung}(2017)}]{Yuan_2017}%
  \BibitemOpen
  \bibfield  {author} {\bibinfo {author} {\bibfnamefont {H.}~\bibnamefont {Yuan}}\ and\ \bibinfo {author} {\bibfnamefont {C.-H.~F.}\ \bibnamefont {Fung}},\ }\href {\doibase 10.1088/1367-2630/aa874c} {\bibfield  {journal} {\bibinfo  {journal} {New Journal of Physics}\ }\textbf {\bibinfo {volume} {19}},\ \bibinfo {pages} {113039} (\bibinfo {year} {2017})}\BibitemShut {NoStop}%
\bibitem [{\citenamefont {Fuchs}\ and\ \citenamefont {Caves}(1994)}]{PhysRevLett.73.3047}%
  \BibitemOpen
  \bibfield  {author} {\bibinfo {author} {\bibfnamefont {C.~A.}\ \bibnamefont {Fuchs}}\ and\ \bibinfo {author} {\bibfnamefont {C.~M.}\ \bibnamefont {Caves}},\ }\href {\doibase 10.1103/PhysRevLett.73.3047} {\bibfield  {journal} {\bibinfo  {journal} {Phys. Rev. Lett.}\ }\textbf {\bibinfo {volume} {73}},\ \bibinfo {pages} {3047} (\bibinfo {year} {1994})}\BibitemShut {NoStop}%
\bibitem [{\citenamefont {Winter}(2001)}]{Winter_2001}%
  \BibitemOpen
  \bibfield  {author} {\bibinfo {author} {\bibfnamefont {A.}~\bibnamefont {Winter}},\ }\href {\doibase 10.1088/0305-4470/34/35/333} {\bibfield  {journal} {\bibinfo  {journal} {Journal of Physics A: Mathematical and General}\ }\textbf {\bibinfo {volume} {34}},\ \bibinfo {pages} {7095} (\bibinfo {year} {2001})}\BibitemShut {NoStop}%
\bibitem [{\citenamefont {Li}\ \emph {et~al.}(2015)\citenamefont {Li}, \citenamefont {Pereira},\ and\ \citenamefont {Plosker}}]{LI2015158}%
  \BibitemOpen
  \bibfield  {author} {\bibinfo {author} {\bibfnamefont {J.}~\bibnamefont {Li}}, \bibinfo {author} {\bibfnamefont {R.}~\bibnamefont {Pereira}}, \ and\ \bibinfo {author} {\bibfnamefont {S.}~\bibnamefont {Plosker}},\ }\href {\doibase https://doi.org/10.1016/j.laa.2015.08.037} {\bibfield  {journal} {\bibinfo  {journal} {Linear Algebra and its Applications}\ }\textbf {\bibinfo {volume} {487}},\ \bibinfo {pages} {158} (\bibinfo {year} {2015})}\BibitemShut {NoStop}%
\bibitem [{\citenamefont {Jozsa}(1994)}]{doi:10.1080/09500349414552171}%
  \BibitemOpen
  \bibfield  {author} {\bibinfo {author} {\bibfnamefont {R.}~\bibnamefont {Jozsa}},\ }\href {\doibase 10.1080/09500349414552171} {\bibfield  {journal} {\bibinfo  {journal} {Journal of Modern Optics}\ }\textbf {\bibinfo {volume} {41}},\ \bibinfo {pages} {2315} (\bibinfo {year} {1994})}\BibitemShut {NoStop}%
\bibitem [{\citenamefont {Janzing}(2009)}]{Janzing2009}%
  \BibitemOpen
  \bibfield  {author} {\bibinfo {author} {\bibfnamefont {D.}~\bibnamefont {Janzing}},\ }\enquote {\bibinfo {title} {Entropy of entanglement},}\ in\ \href {\doibase 10.1007/978-3-540-70626-7_66} {\emph {\bibinfo {booktitle} {Compendium of Quantum Physics}}},\ \bibinfo {editor} {edited by\ \bibinfo {editor} {\bibfnamefont {D.}~\bibnamefont {Greenberger}}, \bibinfo {editor} {\bibfnamefont {K.}~\bibnamefont {Hentschel}}, \ and\ \bibinfo {editor} {\bibfnamefont {F.}~\bibnamefont {Weinert}}}\ (\bibinfo  {publisher} {Springer Berlin Heidelberg},\ \bibinfo {address} {Berlin, Heidelberg},\ \bibinfo {year} {2009})\ pp.\ \bibinfo {pages} {205--209}\BibitemShut {NoStop}%
\bibitem [{\citenamefont {\ifmmode~\dot{Z}\else \.{Z}\fi{}yczkowski}\ and\ \citenamefont {Sommers}(2005)}]{PhysRevA.71.032313}%
  \BibitemOpen
  \bibfield  {author} {\bibinfo {author} {\bibfnamefont {K.}~\bibnamefont {\ifmmode~\dot{Z}\else \.{Z}\fi{}yczkowski}}\ and\ \bibinfo {author} {\bibfnamefont {H.-J.}\ \bibnamefont {Sommers}},\ }\href {\doibase 10.1103/PhysRevA.71.032313} {\bibfield  {journal} {\bibinfo  {journal} {Phys. Rev. A}\ }\textbf {\bibinfo {volume} {71}},\ \bibinfo {pages} {032313} (\bibinfo {year} {2005})}\BibitemShut {NoStop}%
\bibitem [{\citenamefont {Hayden}\ \emph {et~al.}(2006)\citenamefont {Hayden}, \citenamefont {Leung},\ and\ \citenamefont {Winter}}]{Hayden2006}%
  \BibitemOpen
  \bibfield  {author} {\bibinfo {author} {\bibfnamefont {P.}~\bibnamefont {Hayden}}, \bibinfo {author} {\bibfnamefont {D.~W.}\ \bibnamefont {Leung}}, \ and\ \bibinfo {author} {\bibfnamefont {A.}~\bibnamefont {Winter}},\ }\href {\doibase 10.1007/s00220-006-1535-6} {\bibfield  {journal} {\bibinfo  {journal} {Communications in Mathematical Physics}\ }\textbf {\bibinfo {volume} {265}},\ \bibinfo {pages} {95} (\bibinfo {year} {2006})}\BibitemShut {NoStop}%
\bibitem [{\citenamefont {Fukuda}(2014)}]{Fukuda2014}%
  \BibitemOpen
  \bibfield  {author} {\bibinfo {author} {\bibfnamefont {M.}~\bibnamefont {Fukuda}},\ }\href {\doibase 10.1007/s00220-014-2101-2} {\bibfield  {journal} {\bibinfo  {journal} {Communications in Mathematical Physics}\ }\textbf {\bibinfo {volume} {332}},\ \bibinfo {pages} {713} (\bibinfo {year} {2014})}\BibitemShut {NoStop}%
\bibitem [{\citenamefont {Hayashi}(2017)}]{Hayashi2017}%
  \BibitemOpen
  \bibfield  {author} {\bibinfo {author} {\bibfnamefont {M.}~\bibnamefont {Hayashi}},\ }\enquote {\bibinfo {title} {Entanglement and locality restrictions},}\ in\ \href {\doibase 10.1007/978-3-662-49725-8_8} {\emph {\bibinfo {booktitle} {Quantum Information Theory: Mathematical Foundation}}}\ (\bibinfo  {publisher} {Springer Berlin Heidelberg},\ \bibinfo {address} {Berlin, Heidelberg},\ \bibinfo {year} {2017})\ pp.\ \bibinfo {pages} {357--490}\BibitemShut {NoStop}%
\bibitem [{\citenamefont {Chakraborty}(2008)}]{Santanu_2008}%
  \BibitemOpen
  \bibfield  {author} {\bibinfo {author} {\bibfnamefont {S.}~\bibnamefont {Chakraborty}},\ }\href {https://www.pvamu.edu/mathematics/wp-content/uploads/sites/49/Chakraborty-MAA-R14-071406-Final-_4_-6-12-08.pdf} {\bibfield  {journal} {\bibinfo  {journal} {Applications and Applied Mathematics}\ }\textbf {\bibinfo {volume} {3}},\ \bibinfo {pages} {42} (\bibinfo {year} {2008})}\BibitemShut {NoStop}%
\bibitem [{\citenamefont {Kukulski}\ \emph {et~al.}(2021)\citenamefont {Kukulski}, \citenamefont {Nechita}, \citenamefont {Pawela}, \citenamefont {Puchała},\ and\ \citenamefont {Życzkowski}}]{10.1063/5.0038838}%
  \BibitemOpen
  \bibfield  {author} {\bibinfo {author} {\bibfnamefont {R.}~\bibnamefont {Kukulski}}, \bibinfo {author} {\bibfnamefont {I.}~\bibnamefont {Nechita}}, \bibinfo {author} {\bibfnamefont {L.}~\bibnamefont {Pawela}}, \bibinfo {author} {\bibfnamefont {Z.}~\bibnamefont {Puchała}}, \ and\ \bibinfo {author} {\bibfnamefont {K.}~\bibnamefont {Życzkowski}},\ }\href {\doibase 10.1063/5.0038838} {\bibfield  {journal} {\bibinfo  {journal} {Journal of Mathematical Physics}\ }\textbf {\bibinfo {volume} {62}},\ \bibinfo {pages} {062201} (\bibinfo {year} {2021})}\BibitemShut {NoStop}%
\bibitem [{\citenamefont {Eisert}\ \emph {et~al.}(2000)\citenamefont {Eisert}, \citenamefont {Jacobs}, \citenamefont {Papadopoulos},\ and\ \citenamefont {Plenio}}]{PhysRevA.62.052317}%
  \BibitemOpen
  \bibfield  {author} {\bibinfo {author} {\bibfnamefont {J.}~\bibnamefont {Eisert}}, \bibinfo {author} {\bibfnamefont {K.}~\bibnamefont {Jacobs}}, \bibinfo {author} {\bibfnamefont {P.}~\bibnamefont {Papadopoulos}}, \ and\ \bibinfo {author} {\bibfnamefont {M.~B.}\ \bibnamefont {Plenio}},\ }\href {\doibase 10.1103/PhysRevA.62.052317} {\bibfield  {journal} {\bibinfo  {journal} {Phys. Rev. A}\ }\textbf {\bibinfo {volume} {62}},\ \bibinfo {pages} {052317} (\bibinfo {year} {2000})}\BibitemShut {NoStop}%
\bibitem [{\citenamefont {Wu}\ \emph {et~al.}(2023)\citenamefont {Wu}, \citenamefont {Matsui}, \citenamefont {Forrer}, \citenamefont {Soeda}, \citenamefont {Andr{\'{e}}s-Mart{\'{i}}nez}, \citenamefont {Mills}, \citenamefont {Henaut},\ and\ \citenamefont {Murao}}]{Wu2023entanglement}%
  \BibitemOpen
  \bibfield  {author} {\bibinfo {author} {\bibfnamefont {J.-Y.}\ \bibnamefont {Wu}}, \bibinfo {author} {\bibfnamefont {K.}~\bibnamefont {Matsui}}, \bibinfo {author} {\bibfnamefont {T.}~\bibnamefont {Forrer}}, \bibinfo {author} {\bibfnamefont {A.}~\bibnamefont {Soeda}}, \bibinfo {author} {\bibfnamefont {P.}~\bibnamefont {Andr{\'{e}}s-Mart{\'{i}}nez}}, \bibinfo {author} {\bibfnamefont {D.}~\bibnamefont {Mills}}, \bibinfo {author} {\bibfnamefont {L.}~\bibnamefont {Henaut}}, \ and\ \bibinfo {author} {\bibfnamefont {M.}~\bibnamefont {Murao}},\ }\href {\doibase 10.22331/q-2023-12-05-1196} {\bibfield  {journal} {\bibinfo  {journal} {{Quantum}}\ }\textbf {\bibinfo {volume} {7}},\ \bibinfo {pages} {1196} (\bibinfo {year} {2023})}\BibitemShut {NoStop}%
\bibitem [{\citenamefont {Cuomo}\ \emph {et~al.}(2020)\citenamefont {Cuomo}, \citenamefont {Caleffi},\ and\ \citenamefont {Cacciapuoti}}]{https://doi.org/10.1049/iet-qtc.2020.0002}%
  \BibitemOpen
  \bibfield  {author} {\bibinfo {author} {\bibfnamefont {D.}~\bibnamefont {Cuomo}}, \bibinfo {author} {\bibfnamefont {M.}~\bibnamefont {Caleffi}}, \ and\ \bibinfo {author} {\bibfnamefont {A.~S.}\ \bibnamefont {Cacciapuoti}},\ }\href {\doibase https://doi.org/10.1049/iet-qtc.2020.0002} {\bibfield  {journal} {\bibinfo  {journal} {IET Quantum Communication}\ }\textbf {\bibinfo {volume} {1}},\ \bibinfo {pages} {3} (\bibinfo {year} {2020})},\ \Eprint {http://arxiv.org/abs/https://ietresearch.onlinelibrary.wiley.com/doi/pdf/10.1049/iet-qtc.2020.0002} {https://ietresearch.onlinelibrary.wiley.com/doi/pdf/10.1049/iet-qtc.2020.0002} \BibitemShut {NoStop}%
\bibitem [{\citenamefont {Andres-Martinez}\ \emph {et~al.}(2023)\citenamefont {Andres-Martinez}, \citenamefont {Forrer}, \citenamefont {Mills}, \citenamefont {Wu}, \citenamefont {Henaut}, \citenamefont {Yamamoto}, \citenamefont {Murao},\ and\ \citenamefont {Duncan}}]{andresmartinez2023distributing}%
  \BibitemOpen
  \bibfield  {author} {\bibinfo {author} {\bibfnamefont {P.}~\bibnamefont {Andres-Martinez}}, \bibinfo {author} {\bibfnamefont {T.}~\bibnamefont {Forrer}}, \bibinfo {author} {\bibfnamefont {D.}~\bibnamefont {Mills}}, \bibinfo {author} {\bibfnamefont {J.-Y.}\ \bibnamefont {Wu}}, \bibinfo {author} {\bibfnamefont {L.}~\bibnamefont {Henaut}}, \bibinfo {author} {\bibfnamefont {K.}~\bibnamefont {Yamamoto}}, \bibinfo {author} {\bibfnamefont {M.}~\bibnamefont {Murao}}, \ and\ \bibinfo {author} {\bibfnamefont {R.}~\bibnamefont {Duncan}},\ }\href@noop {} {\enquote {\bibinfo {title} {Distributing circuits over heterogeneous, modular quantum computing network architectures},}\ } (\bibinfo {year} {2023}),\ \Eprint {http://arxiv.org/abs/2305.14148} {arXiv:2305.14148 [quant-ph]} \BibitemShut {NoStop}%
\bibitem [{\citenamefont {Lin}(1991)}]{61115}%
  \BibitemOpen
  \bibfield  {author} {\bibinfo {author} {\bibfnamefont {J.}~\bibnamefont {Lin}},\ }\href {\doibase 10.1109/18.61115} {\bibfield  {journal} {\bibinfo  {journal} {IEEE Transactions on Information Theory}\ }\textbf {\bibinfo {volume} {37}},\ \bibinfo {pages} {145} (\bibinfo {year} {1991})}\BibitemShut {NoStop}%
\end{thebibliography}%

\onecolumngrid
\appendix
\section{Fidelity distribution of random bipartite state under the Haar measure with a separable reference state}
\label{apxA}

%\subsection{}
For {\ab any} Haar uniform random pure state {\ab \(\ket{\psi}\)}, its fidelity with a fixed target state $\ket{\phi}$ is a random variable, and its probability distribution function (PDF), \textcolor{black}{denoted as ${P}_{\phi}(f)$,} is given by Eq.~(\ref{def_pdf}): %{\ab why is the symbol \(\sim\) used?}
\begin{align}
%\cancel{{P}(f|\ket{\phi})}
{\color{black}{P}_{\phi}(f)}
&=\int_{\Psi}\dd\psi\hspace{4pt}\delta\left(f- |\bra{\psi}\ket{\phi}|^2\right)
=\int_{U}\dd U \hspace{4pt}\delta\left(f- |\bra{0}U\ket{\phi}|^2\right)\label{def_pdf},
\end{align}
where U is one of a set of Haar uniform unitary operators. From the unitary invariant properties of the Haar measure, we can see that the PDF is independent of the target state \textcolor{black}{$|\phi\rangle$}. %\sout{since $V$ is an arbitrary unitary.}
\begin{align}
%\cancel{{P}(f|\ket{\phi})}
{\color{black}{P}_{\phi}(f)} = \int_{U'}\dd U' \hspace{4pt}\delta\left(f- |\bra{0}UV\ket{\phi}|^2\right) = %\cancel{{P}(f|V\ket{\phi}}
{\color{black}{P}_{V\ket{\phi}}(f)}
\end{align}
where $V$ is an arbitrary unitary and \(U'=UV\).
Therefore from now on, we take the $\ket{0}$ as our target state.

We know from reference~\cite{PhysRevA.71.032313}, that for Haar uniformly distributed random pure states of dimension \(d\), the PDF of their fidelity \textcolor{black}{$P(f)$} with respect to any fixed reference state is given by Eq. (\ref{haar_pdf}):
\begin{align}
{P}(f) = (d-1)(1-f)^{d-2}.\label{haar_pdf}
\end{align}
Since we consider $\ket{0}$ as the target state, the fidelity in this scenario is equivalent to the probability of getting outcome $\ket{0}$ from measurement, which is a random variable. %{\ab not clear}

Then we extend this problem to a special case, in which we choose a subset of Haar uniform random pure states defined by a fixed Schmidt coefficient set $\boldsymbol{s}=\{s_1,s_2,...,s_d\}$, with $s_i\geq s_j\hspace{4pt}\forall \ i\geq j$. 
\textcolor{black}{The PDF of fidelity conditional on a fixed Schmidt-coefficient vector $\boldsymbol{s}$, \textcolor{black}{denoted as $P_{\phi}(f|\boldsymbol{s})$,} can be obtained by the integral of the delta function of fidelity over the Haar measure normalized within the set $\mathbb{S}_{\boldsymbol{s}}$,
\begin{equation}
    P_{\phi}(f|\boldsymbol{s})
    =
    \frac{
        \int_{Haar} d\psi \,\delta(f-|\braket{\psi}{\phi}|^2)
        \,\delta(S(\ket{\psi})-\boldsymbol{s})
    }{
        \int_{Haar} d\psi \,\delta(S(\psi)-\boldsymbol{s})
    }
    =:
    \int_{Haar(\mathbb{S}_{\boldsymbol{s}})} d\psi \,\delta(f-|\braket{\psi}{\phi}|^2)
    ,
\end{equation}
where $S(\ket{\psi})$ is the Schmidt coefficient of the state $\ket{\psi}$ and the set of quantum states $\mathbb{S}_{\boldsymbol{s}}$ with the same Schmidt coefficient $\boldsymbol{s}$ is defined in Eq. \ref{eq::s_coeff_set}. 
}
Therefore, the conditional PDF can be written as %\sout{an} 
{\ab in} Eq. (\ref{ab_pdf}). %{\ab unitaries are denoted by \(u\) as well as \(U\), kindly check}
\begin{align}
%\cancel{{P}_{s,\ket{\phi}}(f|s)}
{\color{black}{P}_{\phi}(f|s)}
=\iint_{U_a,U_b}\dd U_a \dd U_b\hspace{4pt}\delta\left(f- |\langle{I_s}|(U_a\otimes U_b)\ket{\phi}|^2\right) \label{ab_pdf}
\end{align}
where, ${|I_s\rangle} \equiv  \sum_i s_i \ket{i_a,i_b}$ is a fixed state with the Schmidt coefficients $s$.
Since local unitary transformation does not change Schmidt coefficients of a bipartite pure state, we can see from Eq. (\ref{ab_pdf_1}) that the PDF of fidelity only depends on the Schmidt coefficients of the target state.%\sout{, since $V_a\otimes V_b$ is arbitrary local unitary.}
%\begin{align}
%\cancel{{P}_{s}(f|s,\ket{\phi}) =
%{P}_{s}(f|s) =
%\iint_{U'_a,U'_b}\dd U'_a \dd U'_b\hspace{4pt}\delta\left(f- |\langle{I_s}|(U_a\otimes U_b)(V_a\otimes V_b)\ket{\phi}|^2\right)
%= {P}_{s}(f|s,(V_a\otimes V_b)\ket{\phi}) }\label{ab_pdf_1}
%\end{align}

\begin{align}
%\cancel{{P}_{s}(f|s,\ket{\phi})}
\textcolor{black}{{P}_{\phi}(f|s)} = \iint_{U'_a,U'_b}\dd U'_a \dd U'_b\hspace{4pt}\delta\left(f- |\langle{I_s}|(U_a\otimes U_b)(V_a\otimes V_b)\ket{\phi}|^2\right)
=%\cancel{{P}_{s}(f|s,(V_a\otimes V_b)\ket{\phi})} 
\textcolor{black}{{P}_{(V_a\otimes V_b)\ket{\phi}}(f|s)} \label{ab_pdf_1}
\end{align}
where $V_a\otimes V_b$ is arbitrary local unitary. \textcolor{black}{This gives us the freedom of choice for target states up to local unitaries.}

We are dealing with a scenario where the target state is separable, which means we may choose the target state $\ket{\phi}$ to be $\ket{0,0}$. The fidelity in this scenario corresponds to the probability of obtaining outcomes $\ket{0}$,$\ket{0}$ from two local measurements, \textcolor{black}{say measurement $A$ and measurement $B$}. Consequently, using conditional probability notation, we express this as:
%\begin{align}
%\cancel{{P}_{s} = P(A=0,B=0|I_s) = P(A=0|I_s)P(B=0| A=0,I_s)}%\label{cp}.
%\end{align}
\textcolor{black}{
\begin{align}
|\braket{I_s|U_a\otimes U_b}{\phi}|^2 = p(A=0,B=0) = p(A=0)p(B=0| A=0)\label{cp}.
\end{align}
}
The probability of getting $\ket{0}$ from local measurement at A, with a given $U_a$ and $U_b$ is:
\begin{align}
p_a &\equiv %\cancel{p(A=0|I_s,U_a,U_b)}
{\color{black}p(A=0)} = \bra{0}tr_B\left((U_a\otimes U_b)|{I_s}\rangle\langle{I_s}|(U_a^\dag\otimes U_b^\dag)\right)\ket{0} = \sum_{i}s_i^2|\bra{i}U_a\ket{0}|^2\label{pa}
\end{align}
On the other hand, for given $U_a$ and $U_b$, if the measurement output at A is $\ket{0}$, then the probability of getting $\ket{0}$ from B would be:
\begin{align}
p_b \equiv %\cancel{p(B=0|I_s,A=0,U_a,U_b)}
{\color{black}p(B=0|A=0)} = |\bra{0}U_b\ket{\psi_{A=0}}|^2, \hspace{0.5cm} \text{with} \ket{\psi_{A=0}}\equiv\frac{(\bra{0}U_a\otimes\mathbb{I})|{I_s}\rangle}{\sqrt{p(A=0|I_s,U_a)}}
\end{align}
It can be verified that %\cancel{$p_{ab} = p_a p_b$}
\textcolor{black}{$|\braket{I_s|U_a\otimes U_b}{\phi}|^2 = p_a p_b$}, allowing us to reformulate Eq. (\ref{ab_pdf}) in terms of this representation as shown in Eq. (\ref{papb}):
\begin{align}
%\cancel{{P}_{s}(f|s)}
{\color{black}{P}_{\phi}(f|s)}
=&\iint_{U_a,U_b}\dd U_a \dd U_b\hspace{4pt}\delta(f - {\color{black}p_a p_b})\notag\\
=&\int_{0}^1 \dd p_b\int_{U_a}\dd U_a\hspace{4pt}\left[\int_{U_b}\dd U_b\hspace{4pt}\delta\left(p_b - |\bra{0}U_b\ket{\psi_{A=0}}|^2\right)\right]\delta(f - p_a p_b)\label{papb}
\end{align}
Subsequently, integrating over $U_b$ yields the PDF of fidelity for a pure state, leading to the expression in Eq. (\ref{def_pdf}). Consequently, the PDF of the fidelity distribution with a separable target state is given by Eq. (\ref{final}):
\begin{align}
%\cancel{{P}_{s}(f|s)} 
{\color{black}{P}_{\phi}(f|s)}
=&\int_{0}^1 \dd p_b\int_{U_a}\dd U_a\hspace{4pt}(d-1)(1-p_b)^{d-2}\delta\left(f - p_b\left(\sum_{i}s_i^2|\bra{i}U_a\ket{0}|^2\right)\right)\notag\\
=&\int_{0}^1 \dd p_b\int_{\Psi_a}\dd \psi_a\hspace{4pt}(d-1)(1-p_b)^{d-2}\delta\left(f - p_b\left(\sum_{i}s_i^2|\bra{i}\ket{\psi_a}|^2\right)\right)\label{final}
\end{align}

%\clearpage
\begin{prop}
For dimension d, if the target state is separable and the source states are \textcolor{black}{Haar random} maximally entangled \textcolor{black}{states}, then the PDF of fidelity is given by:
\begin{align}
%{P}_{d}(f) 
%\cancel{{P}(f|d)}
\textcolor{black}{P_{\text{sep}}(f|\text{max. ent.})}= d(d-1)(1-df)^{d-2}\hspace{16pt}f\in\,[0,\frac{1}{d}\,]
\end{align}
\end{prop}
\begin{proof}
Since the source state is maximumly tangled, i.e. $s_i^2 = \frac{1}{d}$ for all i, then by the Eq. (\ref{final}).
\begin{align}
%{P}_{d}(f) 
%\cancel{{P}(f|d)}
\textcolor{black}{P_{\text{sep}}(f|\text{max. ent.})}&= \int_{0}^1 \dd p_b\int_{\Psi_a}\dd \psi_a\hspace{4pt}(d-1)(1-p_b)^{d-2}\delta\left(f - \frac{p_b}{d}\left(\sum_{i}|\bra{i}\ket{\psi_a}|^2\right)\right)\notag\\
&=\int_{0}^1 \dd p_b(d-1)(1-p_b)^{d-2}\delta\left(f - \frac{p_b}{d}\right)
= d(d-1)(1-df)^{d-2}\hspace{16pt}f\in\,\left[0,\frac{1}{d}\,\right]
\end{align}
\end{proof}

\begin{prop}
For qubits (d=2), let \textcolor{black}{the target state be separable, and} the Schmidt coefficient of random sources states \textcolor{black}{be} $\left(\sqrt{\frac{1+\gamma}{2}},\sqrt{\frac{1-\gamma}{2}}\right)$, with a parameter $\gamma\in\,[0,1\,]$, then the PDF of fidelity conditional on $\gamma$ is given by Eq. (\ref{2_pdf}).
\begin{align}
%{P}_\gamma(f) 
%\cancel{{P}(f|\gamma)}
\textcolor{black}{P_{\text{sep}}(f|\gamma)}&=  \begin{cases} \frac{1}{\gamma}\ln{\frac{1+\gamma}{2f}}\hspace{8pt}f\in\left[\frac{1-\gamma}{2},\frac{1+\gamma}{2}\right]\\
\frac{1}{\gamma}\ln{\frac{1+\gamma}{1-\gamma}}\hspace{8pt}f\in\left[0,\frac{1-\gamma}{2}\right]\end{cases}\label{2_pdf}
\end{align}
\end{prop}
\begin{proof}
\begin{align}
%{P}_\gamma(f) 
%\cancel{{P}(f|\gamma)}
\textcolor{black}{P_{\text{sep}}(f|\gamma)}=&\int_{0}^1 \dd p_b\int_{\Psi_a}\dd\psi_a\hspace{4pt}\delta\left(f - p_b\left(\frac{1+\gamma}{2}|\bra{0}\ket{\psi_a}|^2 + \frac{1-\gamma}{2}\left(1-|\bra{0}\ket{\psi_a}|^2\right)\right)\right)\notag\\
=&\int_{0}^1 \dd p_b \int_{0}^1 \dd p_a\hspace{4pt}\delta\left(f - p_b\left(\gamma p_a+\frac{1-\gamma}{2}\right)\right)
%=\frac{1}{\gamma}\int_{0}^1dp_b\int_{\frac{1-\gamma}{2}p_b}^{\frac{1+\gamma}{2}p_b}dz\frac{1}{p_b}\delta(f-z)\notag\\
%=&\frac{1}{\gamma}\int_{\frac{1-\gamma}{2}}^{\frac{1+\gamma}{2}}dz\int_{\frac{2z}{1+\gamma}}^{1}dp_b\frac{1}{p_b}\delta(f-z) + \int_{0}^{\frac{1-\gamma}{2}}dz\int_{\frac{2z}{1+\gamma}}^{\frac{2z}{1-\gamma}}dp_b\frac{1}{p_b}\delta(f-z)\notag\\
= \begin{cases} \frac{1}{\gamma}\ln{\frac{1+\gamma}{2f}},\hspace{8pt}f\in\left[\frac{1-\gamma}{2},\frac{1+\gamma}{2}\right]\\
\frac{1}{\gamma}\ln{\frac{1+\gamma}{1-\gamma}},\hspace{8pt}f\in\left[0,\frac{1-\gamma}{2}\right]\end{cases}
\end{align}
\end{proof}

\end{document}